\documentclass[12pt]{amsart}
\usepackage{amsthm,amssymb,amsmath}
\usepackage{hyperref}

\usepackage[dvips]{graphicx}
\usepackage{amsfonts}

\usepackage{xypic}
\usepackage{amscd}
\usepackage[all]{xy}
\usepackage{textcomp} 
\usepackage{mathrsfs} 

\def\be{\begin{equation}}
\def\ee{\end{equation}}

\def\>>{\rangle}

\def\bm{{\bf m}}
\def\bz{{\bf z}}
\def\bw{{\bf w}}
\def\bt{{\bf t}}
\def\bv{{\bf v}}

\def\cM{\mathcal M}

\def\cZ{\mathcal Z}

\def\bC{\mathbb C}
\def\bN{\mathbb N}
\def\bZ{\mathbb Z}
\def\btau{\boldsymbol\tau}
\def\bnu{\boldsymbol\nu}

\def\Y1{Y^{(1)}}

\DeclareMathOperator{\End}{End}

\newtheorem{theorem}{\emph {Theorem}}
\newtheorem{conjecture}[theorem]{\emph {Conjecture}}
\newtheorem{lemma}[theorem]{\emph {Lemma}}
\newtheorem{corollary}[theorem]{\emph {Corollary}}
\newtheorem{proposition}[theorem]{\emph {Proposition}}
\newtheorem*{definition}{\emph {Definition}}
\theoremstyle{remark}
\newtheorem{remark}{Remark}
\newtheorem*{example}{Example}

\begin{document}	

\title[Inhomogeneous TASEP with spectral parameters]{Inhomogenous
  Multispecies TASEP on a ring with spectral parameters}   
\date{\today}
\author[L. Cantini]{Luigi Cantini}
\address{LPTM, Universit\'e de Cergy-Pontoise (CNRS UMR
    8089), 
Cergy-Pontoise Cedex, France.}
\date{\today}
\email{luigi.cantini@u-cergy.fr}
\keywords{ Asymmetric Simple Exclusion Process, Yang-Baxter, Schubert Polynomials}
\thanks{This work has been partially supported by CNRS through a
  ``Chaire d'excellence''.} 

\begin{abstract}

We study an inhomogenous multispecies version of the Totally Asymmetric Simple
Exclusion Process (TASEP) on a periodic oriented one dimensional
lattice, which depends on two sets of parameters $(\btau,\bnu)$, attached to the
particles.
After discussing the Yang-Baxter integrability of our model, we study
its (unnormalized) stationary measure.
Motivated by the integrability of the model we introduce a further set of
spectral parameters $\bz$, attached to the sites of the lattice, and we
uncover a remarkable  underlying 
algebraic structure. We provide exact formulas for the 
stationary measure and prove the factorization of the stationary
probability of certain configurations in terms of double Schubert
polynomials in $(\btau,\bnu)$.

\end{abstract}


\begin{titlepage}

\maketitle

\end{titlepage}

\section{Introduction}

The Asymmetric Simple Exclusion process (ASEP) is a stochastic process
that in the course of the last thirty years has gained the status of 
a paradigmatic model in the theory of far from equilibrium low
dimensional systems \cite{derrida1998exactly}. 
The model describes the stochastic evolution of particles that occupy
the sites of a one dimensional lattice under the exclusion condition,
which means that each site can contain one particle at most. 
The dynamics involves jumps of the particles on neighboring sites
with asymmetric rates for left or right jumping, modeling in this
way the presence of an external driving force.

On one side the ASEP displays a rich phenomenology and has found a wide
range of applications, going from the study of traffic flow, to that
of surface growth, or sequence alignment (see \cite{chou2011non} for a
recent review of several of these applications).  
On the other side the ASEP is amenable to a variety of mathematical
approaches, in part leading to complementary sets of results.  Among
these we mention Bethe Ansatz \cite{gwa1992bethe}, quadratic algebras
\cite{essler1996representations}, combinatorics
\cite{corteel2007tableaux}, orthogonal polynomials
\cite{uchiyama2004asymmetric}, random matrices 
\cite{johansson2000shape}, stochastic differential equations
\cite{corwin2012kardar} and hydrodynamic limits
\cite{rezakhanlou1991hydrodynamic}.
 
In the present and in the following companion paper
\cite{cantini-multispecies-2}  we  study a multispecies generalization
of ASEP on a ring, i.e. on a periodic oriented one dimensional lattice. 
In this model, particles belong to different species (labeled by
integers) and 
the exclusion condition is implemented by requiring each lattice site to be
occupied by exactly one particle (at wish one can interpret particles
of a given species as empty sites).
The time evolution consists of swaps of neighboring particles: a
particle of species $\alpha$ on the left swaps its position with a
particle of species $\beta$ on the right with transition rates
$r_{\alpha,\beta}$ given by  
$$
r_{\alpha,\beta}=\left\{ 
\begin{array}{cc}
  \tau_\alpha-\nu_\beta & \alpha<\beta\\
0 &\alpha\geq\beta
\end{array}\right.
$$
for some family of parameters $\btau=\{\tau_\alpha\}_{\alpha\in \mathbb
  Z}, \bnu=\{\nu_\alpha\}_{\alpha\in\mathbb Z}$.
In particular, since particles of higher species cannot overcome particles of
lower species, we speak of a multispecies  Totally Asymmetric
Exclusion process (M-TASEP).
As will be explained in Section \ref{integr-section}, such a choice
ensure the Yang-Baxter integrability of the M-TASEP.
Our main focus in the present paper will be on the \emph{stationary
probability}, and moreover we will restrict to a system with a single
particle per species on
a ring of length $N$.  On one hand this allows to use a
light notation, on the other hand, as it will explained in
\cite{cantini-multispecies-2} the results for more general species content 
can be derived from the ones presented here. 

For some choices of the parameters $\btau,\bnu$ the model has already
appeared in the literature. The case $\nu_\alpha=\tau_\alpha$ has been
considered by Karimipour in \cite{karimipour1999multispecies}, where
using a matrix product representation he showed that the stationary
measure is uniform.
Another case has appeared in the work of R\'akos and Sch\"utz
\cite{rakos2005bethe}. They considered a system of $N$ species of
particles, each species moving to the right on empty sites with rates
$v_\alpha$, but exchange of particles is
forbidden and since the particles cannot exchange position, one can assume
each particle to be of a different species. In order to fit this model
in our framework we identify particles of species
$N+1$ as empty sites and to forbid the exchange of particles of
successive species $\alpha$, $\alpha+1$, which means
$\nu_{N+1}=0$, and  $\nu_{\alpha}=\tau_{\alpha-1}$ for
$\alpha\leq N$. 

The main motivation of the present work come though from yet
another particular case  considered by Lam and Williams
\cite{lam2012markov} in which all the parameters $\nu_\alpha$ vanish. 
Lam and Williams conjectured that the stationary probabilities of the
particles configurations, apart for a normalization factor called \emph{partition function}, turn out
to be polynomials in the parameters $\btau =\{\tau_\alpha\}$, with
positive integer coefficients. Actually they made an even stronger and
more intriguing conjecture, 
namely that  the unnormalized probability $\psi_\bw(\btau)$ of any
particle configuration $\bw$ is a non negative integral
sum of \emph{Schubert polynomials} in the variables $\btau$. On top of
this they gave explicit formulas for certain components as products of Schubert
polynomials \cite[Conjecture 3 and 4]{lam2012markov}(see Appendix 
\ref{app-schubert} for the definition of Schubert polynomials).  
The weaker result on integrality and positivity of  the coefficients
of $\psi_\bw(\btau)$
was soon settled in two steps. As a first step
Ayyer and Linusson \cite{ayyer2014inhomogeneous} 
gave a conjectural combinatorial expression of the
integers coefficients as enumerating certain multiline queues
previously introduced by Ferrari and Martin
\cite{ferrari2005multiclass,ferrari2007stationary}. Shortly 
later Arita and Mallick  \cite{arita2013matrix} proved Ayyer-Linusson conjecture by constructing a
matrix product ansatz representation of $\psi_\bw(\btau)$
which turns out to be equivalent to the multiline queues. 
Since then, many known results about  stationary measure of the
M-TASEP have been
obtained by using the multiline queues (see for example
\cite{aas2015continuous,aas2013product,ayyer2014correlations}). 
Still the approach through multiline queues has given no insight for
explaining the appearance of Schubert polynomials in this problem.
Moreover the matrix representation of multispecies TASEP has been
rederived recently in the framework of the Zamolodchikov tetrahedron
equation \cite{kuniba2015multispecies, kuniba2016multispecies}.

 \vskip .3cm

 For generic $\bnu$, some of Lam and Williams
conjectures extend  in a natural way. It is convenient to think at the
unnormalized probabilities $\psi_\bw(\btau,\bnu)$  as components of a
vector $\Psi_N(\btau,\bnu)$ in a basis labeled by particle
configurations $\bw$, therefore in this paper we speak of
``components'' instead of ``unnormalized probabilities''.
Let's specify the normalization of $\Psi_N(\btau,\bnu)$ by fixing the
component associated to the configuration
$
12\dots N
$
as
\be\label{norm-general}
\psi_{12\dots N}(\btau,\bnu) =  \phi_{N}(\btau,\bnu),
\ee
with
\be\label{def_norm-general}
\phi_{N}(\btau,\bnu) :=
\prod_{1\leq \alpha< \beta \leq N}(\tau_\alpha-\nu_\beta)^{\beta-\alpha-1}.
\ee
We have the following Theorem that shall be proven in the paper (it
will be a corollary of Theorem \ref{primality-general})
\begin{theorem}\label{theo-minimal-poly}
With the normalization given by
eq.(\ref{norm-general}), the components $\psi_{\bw}(\btau,\bnu)$  are 
relative prime polynomials in $\btau,\bnu$, with integer coefficients.
\end{theorem}
Moreover numerical computations at small sizes suggest the following
\begin{conjecture}\label{conj-positivity}
With the normalization given by
eq.(\ref{norm-general}), the components $\psi_{\bw}(\btau,-\bnu)$  are
polynomials in $\btau,\bnu$, with positive integer coefficients.
%
\end{conjecture}
A natural question to ask is whether these coefficients have any 
combinatorial origin. This amounts to ask whether there exist
combinatorial objects, possibly generalizing the multiline queues 
of Ferrari and Martins, whose appropriately weighted enumerations
coincide with $\psi_\bw(\btau,-\bnu)$.
A related question is to construct a Matrix Product Ansatz \cite{derrida1993exact}
representation of the stationary measure.

In this paper we study the stationary measure through an approach
which is based on ideas introduced by Di Francesco and Zinn-Justin
in the context of the stochastic dense O(1) loop model
\cite{di2005around}. For a system on a ring of length $L$ we deform
the master equation for the stationary measure by introducing
scattering matrices that depends on $L$ spectral parameters
$\bz=\{z_1, \dots,z_L\}$. The scattering matrices have a common stationary
state $\Psi_N(\bz)$ that reduces to $\Psi_N(\btau,\bnu)$ for
$z_i=\infty$. We show that $\Psi_N(\bz)$ is solution of 
 a set of exchange equations. Such equations
involve certain divided difference operators $\pi_i(\alpha,\beta)$
(see eq.(\ref{divided-ab}) for their definition) that
generalize the isobaric divided difference operators of Lascoux and
Sch\"utzenberger and whose commutation relation generalize to one
satisfied by the generators of the $0-$Hecke algebra \cite{lascoux1983symmetry,fomin1994grothendieck}. 

By analyzing the exchange equations we
prove exact expressions for the components associated to several
configurations. Notably we show that the component $\psi_{N(N-1)\dots
  1}(\bz)$ associated to the configuration 
$$
N(N-1)\dots 21
$$
factorizes in terms of polynomials $\mathfrak S^{r, s}(\bz)$ that
correspond to a $\bz$ deformations of certain Double Schubert
polynomials, thus proving a generalization of one of Lam and Williams
conjectures.      
As a byproduct of the analysis we show that many components
factorize in terms of $\mathfrak S^{r, s}(\bz)$. Moreover, the knowledge of the explicit form of $\psi_{N(N-1)\dots
1}(\bz)$ allows to compute the so called partition function
$\cZ_N(\bz,\btau,\bnu)$, i.e. the sum of all the components and also
to show a remarkable factorization of the stationary measure that has
already been proven for $\nu_\alpha=0$ by Aas and Sj{\"o}strand
\cite{aas2013product} using multiline queues enumerations.

The paper is organized as follows. In Section \ref{notation} we 
discuss briefly multispecies exclusion processes fixing the notations
we use in the rest of the paper. 
In Section \ref{integr-section} we 
discuss the Yang-Baxter integrability of the multispecies exclusion
processes, in particular we show how integrability leads to the
exchange rates discussed above.

Sections \ref{exch-section} and \ref{sect:solution} form the core of
the paper. 
In Section \ref{exch-section} we deform the master equation for the
stationary measure by introducing scattering operators that involve
the spectral parameters. Then unique stationary measure of the 
scattering operators depends on the spectral parameters and reduces
the the stationary measure of the original M-TASEP when
the spectral parameters are set to $0$. In the same section we show
that the stationary measure of the scattering operators can be
normalized in such a way to satisfy certain exchange equations. In
Section \ref{sect:solution} we start by analyzing the exchange
equations, by expressing them in terms of divided difference
operators acting on the configurations probabilities. Then in Section 
\ref{sect:triv-fact} we derive trivial factors of the components, this
allows to compute $\psi_{12\dots N}(\bz)$ and to determine the degree
$\Psi_N(\bz)$ as a polynomial in $\bz$. In Section \ref{sect:recurs}
we derive recursions relating $\Psi_N(\bz)$ and $\Psi_{N-1}(\bz)$.
These will be used in Section \ref{sect:desc-conf} to provide an the
formula for $\psi_{N(N-1)\dots 1}(\bz)$. Then in Section
\ref{sect:normal} we show the factorization property of the stationary
measure and compute the partition function.

In Appendix \ref{app-schubert} we present the definition of
Double Schubert polynomials, while in Appendix \ref{app-proofs} we
gather some technical results used in the paper.

\section{Multispecies exclusion processes} \label{notation}

In \emph{multispecies exclusion processes}
each site of a periodic oriented lattice (a ring) is occupied by a  
particle, and the particles belong to different species labeled by
integers. 
The dynamics takes place in continuous time and consists of local
updates of pairs of neighboring sites: if the site $i$ is occupied by a
particle of specie $\alpha$ and site $i+1$ is occupied by a particle
of specie $\beta$ then the exchange rate is
$r_{\alpha,\beta}$.  
\vskip .7cm
\begin{center}
\includegraphics[scale = .45]{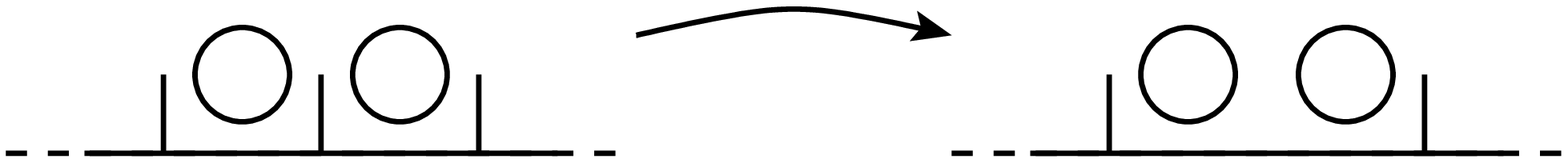}
\end{center}
\begin{picture}(0,0)
\put(170,49){$r_{\alpha,\beta}$}
\put(86,29){$\alpha$}
\put(265,29){$\alpha$}
\put(112,28){$\beta$}
\put(240,28){$\beta$}
\end{picture}

\vskip -.3cm
Particles belonging to the same specie are
considered as indistinguishable, therefore there is no rate
$r_{\alpha,\alpha}$.  

For a periodic lattice of length $L$, a configuration
$\bw=\{w_1w_2\dots w_L\}$ is specified by assigning to 
each site $i$ an  integer $w_i$ corresponding to the specie of the
 particle occupying that site.
Since the dynamics preserves the number of particles, the system is
completely specified when we fix $m_\alpha$, the number of particles
of species $\alpha$ present on the lattice. If there are $N$ different
species of particles on the lattice, we can assume (up to relabeling
of the species index) that $m_\alpha>0$ for $1\leq \alpha\leq N$ and
call $\bm=\{m_1,\dots,m_N\}$ the \emph{species   content}. 
\begin{example}
Here is an example of a configuration for a system of size $L=8$
$$
\bw=\{1,2,2,4,1,5,3,5\},~~~~~\bm(\bw)=\{2,2,1,1,2\}. 
$$
\end{example}
We call $\mathcal H_{\bm}$ the vector space of states of content
$\bm$, with a preferred  basis $v_\bw$ labeled by configurations $\bw$ such that
$\bm(\bw)=\bm$.  In order to write explicitly the Markov matrix
defining the stochastic evolution of the system on $\mathcal
H_{\bm}$, it is convenient to introduce $\mathcal H^{L}$, the space with
unconstrained content, which has a
preferred basis labeled by all the configurations $\bw$ of length
$L(w)=L$. Such space $\mathcal
H^{L}$ has a natural tensor product structure
$$
\mathcal H^L = V_1\otimes \cdots
\otimes V_i\otimes \cdots \otimes V_L  
$$  
where $V_i \simeq V\simeq \bC^\infty$, with a preferred basis
$\{v_\alpha|\alpha\in \bN \}$. The state space $\mathcal H_{\bm}$ of
a M-TASEP of content ${\bf   m}$ is naturally embedded as a subspace
in $\mathcal H^{L}$ with $L=|\bm|:=\sum_{\alpha=1}^Nm_\alpha$.

Call $p_\bw(t)$ the probability of having a the configuration
$\bw$ of content $\bm$ at time $t$. The time evolution of $p_\bw(t)$ is
determined by the Master equation 
$$
\frac{d}{dt}\mathcal P_\bm(t) =  \cM \mathcal P_\bm(t)
$$
where the probabilities $p_\bw(t)$ are gathered in the vector 
$$
\mathcal P_\bm(t):=\sum_{\bw|\bm(\bw)=\bm}
p_\bw(t) v_\bw 
$$
and  the Markov matrix $\cM$ is written as a
sum of local exchange terms 
\be
\cM = \sum_{i=1}^L h_{i} .
\ee
The operator $h\in \End(V\otimes V)$, which accounts for the local
exchange rates on two consecutive sites, reads
\be
h = \sum_{\alpha ,\beta \in \bZ} r_{\alpha,\beta} T^{(\alpha,\beta)}
\ee
with
\be
T^{(\alpha,\beta)} = E^{(\beta,\alpha)} \otimes
E^{(\alpha,\beta)} - E^{(\alpha,\alpha)}\otimes E^{(\beta,\beta)},  
\ee
and the elementary operators $E^{(\alpha,\beta)}\in\End (V)$ act on
the basis $\{v_\gamma\}$ of $V$ by $E^{(\alpha,\beta)} v_\gamma=
\delta_{\beta,\gamma} v_\alpha$.
The operators $h_{i}$ act locally on the tensor product
$V_i\otimes V_{i+1}$
$$
h_{i}= {\bf 1}_1\otimes \cdots {\bf 1}_{i-1}\otimes h \otimes  {\bf
  1}_{i+2} \otimes \cdots \otimes {\bf 1}_L.
$$
In the present paper we will be concerned only with the stationary
measure $\mathcal P_\bm$. For sufficiently generic rates
$r_{\alpha,\beta}$ such measure is unique and is given by  the 
solution of the master equation
\be\label{eq-stationarity}
\cM \mathcal P_\bm=0.
\ee
The approach we will adopt consists in deforming the
previous master equation. In order to this we have first to discuss
the integrability of the M-TASEP.

\section{Integrability}\label{integr-section}

The standard way to show the Yang-Baxter integrability of an operator  
like $\cM$, given by the sum of local operators, is to find  
$\check R_{i}(x,y)$ matrices, acting on $V_i\otimes
V_{i+1}$, 
such that   
\be\label{simple-rel}
\begin{split}
&\check R_{i}(x,x)={\bf 1}\\
&\check R_{i}(x,y)\check R_{i}(y,x)={\bf 1}\\
&\frac{d}{dx} \check R_{i}(x,y)|_{x=y=c} \propto h_{i}
\end{split}
\ee
and such that they satisfy the braided Yang-Baxter equation
\be\label{YBE1}
\check R_{i}(y,z)\check R_{i+1}(x,z)\check R_{i}(x,y)=\check
R_{i+1}(x,y)\check R_{i}(x,z)\check R_{i+1}(y,z). 
\ee
Motivated by the fact that $h$ is itself the sum of more
elementary operators $T^{(\alpha,\beta)}$, we search $\check R$ matrices of
the baxterized form
\be\label{baxterized1}
\check R_{i}(x,y) = {\bf 1} + \sum_{\alpha, \beta \in \bZ}
g_{\alpha,\beta}(x,y) T_{i}^{(\alpha,\beta)}.    
\ee
The following result is probably well-known, but apparently never
stated explicitly
\begin{theorem}
Assume that $\forall \alpha,\beta$  $g_{\alpha,\beta}(x,y)$ do
not vanish identically. Then, up to reparametrization of the spectral
variables $(x,y)$ and reordering of the species labels, the
baxterized solutions of eqs.(\ref{simple-rel},\ref{YBE1}) are labeled
by a parameter $q$ and read
\be\label{hecke-case}
g_{\alpha,\beta}(x,y) = \frac{x-y}{1-(q+q^{-1})y+ xy}
q^{\textrm{sign}(\alpha- \beta)}.  
\ee
\end{theorem}  
\begin{remark} The solution corresponding to (\ref{hecke-case}) is
nothing else than the baxterization of the Hecke algebra, indeed the
operators
\be
E_i = \sum_{\alpha,\beta \in \bZ}q^{\textrm{sign}(\alpha- \beta)} T_{i}^{(\alpha,\beta)}
\ee
satisfy the Hecke relations \cite{alcaraz1993reaction}
\be
\begin{split}
E^2_i&=-(q+q^{-1})E_i\\
[E_i,E_j]&=0~~\textrm{for}~~|i-j|>1\\
E_iE_{i+1}E_i-E_i&=E_{i+1}E_iE_{i+1}-E_{i+1}.
\end{split}
\ee
\end{remark}
In order to obtain a richer family of solutions we set to zero
the function $g_{\alpha,\beta}(x,y)$ for $\alpha>\beta$.
\begin{theorem}\label{Theo-par-sol}
Suppose that $g_{\alpha,\beta}(x,y)=0$ for $\alpha>\beta$, while
$g_{\alpha,\beta}(x,y)$  not identically zero for $\alpha<\beta$. Then
the most
general solution of eqs.(\ref{simple-rel},\ref{YBE1}), of the form
(\ref{baxterized1}), is given for $\alpha<\beta$ by
\be\label{gen-sol-Tasep}
g_{\alpha,\beta}(x,y) = g(x,y|\tau_\alpha,\nu_\beta):=
1-\frac{f(x|\tau_\alpha,\nu_\beta)}{f(y|\tau_\alpha,\nu_\beta)},
\ee
with
$$
f(x|\tau,\nu)=\frac{x-\tau}{x-\nu}. 
$$
\end{theorem}
\begin{proof}
Call $Y\!B$ the difference between left and right
hand side of the Yang-Baxter equations (\ref{YBE1}), then we look for
the solutions of the equations $Y\!B^{\Theta'}_\Theta=0$  with the 
multi-indices of $\Theta=\{\theta_1,\theta_2,\theta_3\}$ and
$\Theta'=\{\theta'_1,\theta'_2,\theta'_3\}$ which are related
by a permutation. 

If we restrict the elements of $\Theta$ to the set
$\{\alpha,\beta\}$, then what we obtain are the Yang-Baxter equations
of a 
problem with just the two species $\alpha$ and $\beta$.
Direct inspection of these equations shows that their solution take the form
$g_{\alpha,\beta}(x,y)=1-\frac{f_{\alpha,\beta}(x)}{f_{\alpha,\beta}(y)}$,
where at  this point   $f_{\alpha,\beta}(x)$ is an arbitrary function.
In order to see how the different functions $f_{\alpha,\beta}(x)$ are
related one has to look at equation in which all the three species
$\alpha,\beta$ and $\gamma$ appear.
Let's look at the equations
$YB^{\{\beta,\gamma,\alpha\}}_{\{\alpha,\beta,\gamma\}}=0$ and
$YB^{\{\gamma,\alpha,\beta\}}_{\{\alpha,\beta,\gamma\}}=0$, which read 
respectively 
\begin{align}\label{abceq1}
(g_{\alpha,\beta}(y,z)g_{\alpha,\gamma}(x,z)-g_{\alpha,\beta}(x,z)g_{\alpha,\gamma}(y,z)) 
(g_{\beta,\gamma}(x,y)-1)&=0,\\\label{abceq2}
(g_{\alpha,\gamma}(x,z)g_{\beta,\gamma}(x,y)-g_{\alpha,\gamma}(x,y)g_{\beta,\gamma}(x,z)) 
(g_{\alpha,\beta}(y,z)-1)&=0.
\end{align}
If $\alpha < \beta,\gamma$ then the functions $g_{\alpha,\beta}(x,y)$
and $g_{\alpha,\gamma}(x,y)$ are not identically zero and one can
rewrite eq.(\ref{abceq1}) as 
\be\label{abceq3}
\frac{g_{\alpha,\beta}(x,z)}{g_{\alpha,\gamma}(x,z)}=
\frac{g_{\alpha,\beta}(y,z)}{g_{\alpha,\gamma}(y,z)}= F^{\alpha,\beta,\gamma}_1(z) 
\ee
for some $F^{\alpha,\beta,\gamma}_1(z)$ which depends only on $z$. Analougously if $\gamma >
\alpha,\beta$, eq.(\ref{abceq2}) can be rewritten as
\be\label{abceq4}
\frac{g_{\alpha,\gamma}(x,y)}{g_{\beta,\gamma}(x,y)}=
\frac{g_{\alpha,\gamma}(x,z)}{g_{\beta,\gamma}(x,z)}=F^{\alpha,\beta,\gamma}_2(x)  
\ee
for some $F^{\alpha,\beta,\gamma}_2(x)$ which depends only on $x$.
Once expressed in terms of $f_{\alpha,\beta}(x), f_{\alpha,\gamma}(x)$
and $f_{\beta,\gamma}(x)$, eqs.(\ref{abceq3},\ref{abceq4}) imply that
these functions are  related  one to the other by projective
transformations. Therefore without loosing generality we  can assume     
$$
f_{\alpha,\beta}(x)=\frac{x-\tau_{\alpha,\beta}}{x-\nu_{\alpha,\beta}}
$$
When this form of $f_{\alpha,\beta}(x), f_{\alpha,\gamma}(x)$
and $f_{\beta,\gamma}(x)$ is plugged into eq.(\ref{abceq3}) one finds
that $\tau_{\alpha,\beta}=\tau_{\alpha,\gamma} $, while from
eq.(\ref{abceq4}) one finds $\nu_{\alpha,\gamma}=\nu_{\beta,\gamma}$
which mean that the parameters $\tau$ and $\nu$ depend 
only on the first and on the second index respectively.
In order to conclude the proof it is sufficient to check that with
the choice $g_{\alpha,\beta}(x,y)$ of eq.(\ref{gen-sol-Tasep}) all the
other components of the Yang-Baxter equations automatically vanish, which is
the case.
\end{proof}
\begin{remark}
Actually with just a little more annoying work one can relax the
hypothesis of Theorem \ref{Theo-par-sol}. It is enough to assume that
for some $\gamma< \delta$, $g_{\delta,\gamma}(x,y)=0$, while
$g_{\gamma,\delta}(x,y)$ not identically zero. Then, up to relabeling
of the species, one can deduce that $g_{\alpha,\beta}(x,y)=0$ for $\alpha>\beta$.
\end{remark}
\begin{remark}
The solution of the Yang-Baxter equation of Theorem \ref{Theo-par-sol}
was first found in an implicit form in \cite[Appendix
  B]{cantini2008algebraic}. Recently such solution has found a nice
algebraic formulation and generalization in \cite{crampe2015new}.
\end{remark}
The derivative of $\check R_{i}(x,y)$ specialized in $x=y=c$ reads
\be
h_{i}(c)=c^2\frac{d}{dx} \check R_{i}(x,y)|_{x=y=c} = \sum_{\gamma
  \leq \alpha<\beta \leq \delta}
\left(\frac{c^2}{\tau_\alpha-c}-\frac{c^2}{\nu_\beta-c} \right )T_{i}^{(\alpha,\beta)}. 
\ee
By setting $c=\infty$ we obtain the exchange rates
\begin{align}\label{rates}
  r_{\alpha \geq \beta}=0,&&r_{\alpha < \beta} = \tau_\alpha-\nu_\beta.&
\end{align}
This is the class of models whose stationary measure we analyze
in the rest of the paper. Certain particular cases of this class have already 
appeared in the literature. 

In the introduction we have already mentioned the work of Lam and
Williams \cite{lam2012markov} about the case $\nu_\alpha=0$,  
which has been one of the main
motivation of the present paper.

In \cite{karimipour1999multispecies} Karimpour studied the case in which $N$ species
of particles moves on a ring 
with empty spaces with the following rules: a particle of type
$\alpha$ moves to an ampty site with rate $v_\alpha$ (which is
interpreted as the ``speed'' of this specie), while two
particles of species $\alpha$ and $\beta$ exchange position with rate
$v_\alpha-v_\beta$ if it is positive or they do not move (without
loosing generality one can assume  $v_\alpha \geq v_{\alpha+1}$. This
model would correspond in our language to a system with $N+1$ species,
the $N+1$st corresponding to empty sites and parameters
$\nu_\alpha=\tau_\alpha=v_\alpha$ for $1\leq \alpha \leq N$ and
$\nu_{N+1}=0$. Using a Matrix Product Ansatz, Karimpour showed that the
stationary probability is simply the uniform measure, this result can
also be easily recovered with our approach that will be explained in
Section \ref{sect:exch-analysis}.

Another case has appeared in the work of R\'akos and Sch\"utz
\cite{rakos2005bethe}. They considered a system of $N$ species of
particles, each species moving to the right on empty sites with rates
$v_\alpha$ as in Karimpour's models, but exchange of particles is
forbidden and since the particles cannot exchange position, one can assume
each particle to be of a different species. In order to fit this model
in our language it is enough to identify particles of species
$N+1$ as empty sites and to forbid the exchange of particles of
successive species $\alpha$, $\alpha+1$, which means
$\nu_{N+1}=0$, and  $\nu_{\alpha}=\tau_{\alpha-1}$ for
$\alpha\leq N$. In this way, the subset of 
configurations in which (up to cyclic permutations) particles with
labels less then $N+1$ are ordered increasingly is absorbing and
preserved by time evolution.

\section{Exchange equations}\label{exch-section}

The goal of this section is to use the integrability to deform the
stationary eqs.(\ref{eq-stationarity}) by introducing the spectral parameters.
The starting point are the scattering matrices $S_{i}({\bf z})$,
defined by 
\be\label{def-scatt}
S_{i}({\bf z}):= \mathcal R \check
R_{i-2}(z_i,z_{i-1}) \dots \check R_{i+1}(z_i,z_{i+2}) \check R_i(z_i,z_{i+1})
\ee
where $\mathcal R \in \End(V_1\otimes V_2\otimes \cdots\otimes
V_N)$ is the operator that ``rotates'' our system
\be
\mathcal R (v_1\otimes v_2 \otimes \cdots \otimes
v_{N-1}\otimes v_N) := (v_N\otimes v_1\otimes v_2 \otimes \cdots \otimes
v_{N-1}).
\ee 
Thanks to the Yang-Baxter equation (\ref{YBE1}) it is easy to verify
that the scattering matrices commute among themselves 
\be
[S_i({\bf z}),S_j({\bf z})]= 0.
\ee 
Moreover we have the following important
\begin{proposition}\label{markov-scatt}
The scattering matrices $S_i({\bf z})$, acting on $\mathcal H_{\bm}$, have a
single common eigenvector $\Psi_\bm(\bz)$ of eigenvalue $1$ for any $i$.
\end{proposition}
\begin{proof}
By choosing $\bz$, $\btau$ and $\bnu$ such that $\forall j\neq i$, 
$0<g_{a,b}(z_i,z_j)<1$ we have that the matrices $S_i(\bz)$ are
irreducible stochastic matrices in any sector $\mathcal H_\bm$, hence
each of them has a single right eigenvector with eigenvalue $1$ in
$\mathcal H_\bm$, that we call $\Psi_i(\bz)$. It remains to show
that   $\Psi_i(\bz)=\Psi_j(\bz)$ for 
$i\neq j$. Suppose by absurd that this is not the case. Since
the matrices $S_i({\bf z})$ commute, all the vectors $\Psi_j(\bz)$ are
right eigenvectors of any  $S_i({\bf z})$.
By absurd therefore one should have that for some
$i\neq j$,  $S_i({\bf z})\Psi_j(\bz)= \lambda_{i,j}\Psi_j(\bz)$
with $\lambda_{i,j}\neq 1$. But this would mean that $\Psi_j(\bz)$ is
orthogonal to the common left eigenvector (whose entries are all
equal), which is impossible by the Perron-Frobenius theorem.
\end{proof}
Since the entries of $S_i({\bf z})$ are rational functions
of the variables ${\bf z}$, $\Psi_\bm({\bf z})$  can be normalized in such
a way that its entries are polynomials of such variables. 
The key result that allows to compute $\Psi_\bm({\bf z})$ is the following
\begin{theorem}[Exchange equations]
Let $\Psi_\bm(\bz)$ be the unique (up to scalar multiplication) common 
eigenvector of $S_i(\bz)$ in the sector $\mathcal H_\bm$, with eigenvalue $1$ and
normalized in such a way that its components are polynomials of $\bz$
of minimal degree. Then $\Psi_\bm(\bz)$ satisfies the following 
exchange equations
\be\label{exchange-eq}
\boxed{\check R_i(z_i,z_{i+1})\Psi_\bm(\bz) = s_i\circ \Psi_\bm(\bz)}
\ee 
where $s_i$ exchange the variables $z_i \leftrightarrow z_{i+1}$.
\end{theorem}
\begin{proof}
Take $j\neq i,i+1$, then from the Yang-Baxter equation (\ref{YBE1}) we
immediately find that 
$$
R_i(z_i,z_{i+1})S_j({\bf z}) = (s_i\circ S_j({\bf z}))R_i(z_i,z_{i+1}).
$$
This means that $R_i(z_i,z_{i+1})\Psi_\bm(\bz)$ is an eigenvector of
$(s_i\circ S_j(\bz))$ with eigenvalue equal to $1$, therefore by
uniqueness  it must be proportional to  $s_i \circ \Psi_\bm(\bz)$.
$$
\check R_i(z_i,z_{i+1})\Psi_\bm({\bf z}) = c_i({\bf z}) s_i \circ \Psi_\bm({\bf z}).
$$ 
The proportionality factor $c_i({\bf z})$ is a 
rational function of ${\bf z}$ and, since $\Psi_\bm({\bf z})$ is supposed to
be of minimal degree, its denominator part can possibly come only from
the poles of $\check R_i(z_i,z_{i+1})$, in particular it depends only on
$z_i,z_{i+1}$ in a factorized form, i.e. $c_i({\bf z}) =
\frac{\bar c_i({\bf z})}{k^{(i)}_1(x_i)k^{(i)}_2(x_{i+1})}$,
where $k^{(i)}_1(x_i)$ is some product of $(x_i-\tau_\alpha)$, while
$k^{(i)}_2(x_{i+1})$  is some product of $(x_i-\nu_\beta)$.  From $\check
R_i(z_i,z_{i+1})\check R_i(z_{i+1},z_{i})= 1$ it follows that
$c_i({\bf z})\left( s_i\circ c_i({\bf z})\right)=1$ and from $\check R_i(z,z)=1$ it
follows $c_i({\bf z})|\{z_i=z_{i+1}\}=1$. Combining these information we
conclude that $c_i({\bf z})$ must be of the form $c_i({\bf
z})=\frac{k^{(i)}_1(x_{i+1})k^{(i)}_2(x_{i})}{k^{(i)}_1(x_i)k^{(i)}_2(x_{i+1})}$. 
Remark that
for any pair $(i,j)$,  $k^{(i)}_{1}(x)$ and
$k^{(j)}_{2}(x)$ have no common factors. 
 
Now specialize $z_1=z$ and for $i\neq 1$, $z_i=w$, by repeatedly
applying the exchange equation it follows that  
$$
\check R_N(z,w)\dots \check R_1(z,w) \Psi_\bm(z,w,\dots,w)= \prod_{i=1}^N
\frac{k^{(i)}_1(w)k^{(i)}_2(z)}{k^{(i)}_1(z)k^{(i)}_2(w)}\Psi_\bm(z,w,\dots,w) 
$$ 
Contracting  both sides of the previous equation with the dual eigenvector we
find that
$$
\prod_{i=1}^N
\frac{k^{(i)}_1(w)k^{(i)}_2(z)}{k^{(i)}_1(z)k^{(i)}_2(w)}=1
$$
which, in view of the absence of common factors between $k^{(i)}_{1}(x)$ and
$k^{(j)}_{2}(x)$, forces all the $k^{(i)}_{\alpha}(x)$ to be equal to $1$.
\end{proof}
\begin{corollary}
The exchange equations (\ref{exchange-eq}) have a unique solution up to
multiplication by a symmetric function of ${\bf z}$.
\end{corollary}
At this point we can easily relate $\Psi_\bm(\bz)$ to
$\mathcal P_\bm$, the stationary measure in the sector
$\mathcal H_\bm$: for $\bz=\infty$ the two are just
proportional\footnote{If $P({\bf x})$ is a
  polynomial 
  of total degree $D$ in the variables ${\bf  x}=x_1,x_2,\dots,x_L$,
  by $P(\infty)$   we mean the coefficient of the monomial of top
  degree of  $P({\bf x})|_{x_1=x_2=\cdots=x_L=x}$,i.e.
$$
P(\infty):=\lim_{x\rightarrow \infty} x^{-D}P({\bf x})|_{x_1=x_2=\cdots,x_L=x}.  
$$
} 
\be\label{specialization1}
\Psi_\bm(\infty) \propto \mathcal P_\bm.  
\ee
Indeed by differentiating eq.(\ref{exchange-eq}) with respect to $z_i$ and then 
setting $\bz=c$ we get 
$$
h_i(c)\Psi_\bm(c)+c^2\partial_i \Psi_\bm(c) = c^2\partial_{i+1} \Psi_\bm(c)
$$ 
from which we see that the sum of the terms $h_i(c)\Psi_\bm(c)$ is
telescopic and at the end we obtain 
\be\label{stat-psi1}
\sum_{i=1}^L h_i(c)\Psi_\bm(c) =0.
\ee

In the rest of the paper we concentrate on $\Psi_N(\bz)$ the solution of the
exchange eq.(\ref{exchange-eq}) in the case 
$m_i=1$ for $1\leq i \leq N$ and zero otherwise.  
We will be able to derive
several properties of $\Psi_N(\bz)$, and through a straightforward
specialization $\bz=\infty$ we will settle some of the questions
about $\Psi_N$ mentioned in the Introduction. 

\section{Solution of the exchange equations}\label{sect:solution}

\subsection{The exchange equations in components} \label{sect:exch-analysis}

As a first step we expand the exchange equations (\ref{exchange-eq})
into the basis $v_\bw$ 
\be
\Psi_\bm(\bz) = \sum_{\ell(\bw)=N}\psi_\bw(\bz) v_\bw.
\ee
The components $\psi_\bw(\bz)$ correspond to a deformation of the
(unnormalized) stationary probabilities of the configurations $\bw$. 
In order to write the exchange equation in a compact way it is
convenient to introduce the natural action of the symmetric group
$\mathcal S_N$ on particles configurations, for $\sigma \in  \mathcal S_N$
$$
\sigma \{w_1,\dots,w_N\}= \{w_{\sigma(1)},\dots,w_{\sigma(N)}\}.
$$ 
Then the exchange equations between positions $(i,i+1)$ become 
\begin{equation}
\label{exch2}
\boxed{\psi_\bw(\bz) = \pi_i(w_{i+1},w_i)
\psi_{s_i\circ \bw}(\bz)\hspace{.4cm}\textrm{if} \hspace{.4cm}w_i<w_{i+1}.} 
\end{equation} 
where $\pi_i(\beta,\alpha)$ are isobaric divided difference operators
in the variables $f(z|\tau_\alpha,\nu_\beta)$ defined by 
\be\label{divided-ab}
\boxed{\begin{split}
\pi_i(\beta,\alpha) G({\bf z}) &= f(z_{i+1}|\tau_\alpha,\nu_\beta)\frac{G({\bf z})-s_i\circ G({\bf
    z})}{f(z_{i}|\tau_\alpha,\nu_\beta)-f(z_{i+1}|\tau_\alpha,\nu_\beta)}\\
&= \frac{(z_{i+1}-\tau_\alpha)(z_i-\nu_\beta)}{\tau_\alpha-\nu_\beta}\frac{G({\bf z})-s_i\circ G({\bf
    z})}{z_{i}-z_{i+1}}
\end{split}}
\ee

It is not difficult to realize that the system of equations
(\ref{exch2}) is cyclic: suppose a component associated to
a configurations $\bw$ is known, all the other components can be
obtained by the action of the operators $\pi_i(\alpha,\beta)$.
\vskip .3cm

Before proceeding further in the analysis of eqs.(\ref{exch2}) we want
to spend a couple of words about the operators
$\pi_i(\beta,\alpha)$. These operators satisfy relations that 
generalize the ones satisfied by the generators of the $0-$Hecke algebra
\be
\begin{split}
\pi_i(\alpha,\beta)\pi_j(\gamma,\delta)&=\pi_j(\gamma,\delta)
\pi_i(\alpha,\beta)\hspace{1cm}|i-j|>1 \\ 
\pi_i(\alpha,\beta)\pi_i(\gamma,\delta)&=-\pi_i(\alpha,\beta)\\
\pi_i(\beta,\gamma)\pi_{i+1}(\alpha,\gamma)\pi_i(\alpha,\beta)&=
\pi_{i+1}(\alpha,\beta)\pi_i(\alpha,\gamma)\pi_{i+1}(\beta,\gamma). 
\end{split} 
\ee
In particular we remark that the last equation is a braided
Yang-Baxter equation. The $0-$Hecke algebra is recovered in the case
we choose $\tau_\alpha=\tau$ and $\nu_\alpha=\nu$ $\forall \alpha$, in
which case the operators $\pi_i(\alpha,\beta)$ become independent of
the labels $\alpha,\beta$ and correspond to the more common isobaric divided
difference operators $\pi_i$ in the variable $\frac{1-\tau x}{1+\nu
  x}$ \cite{lascoux2003symmetric}. On the other hand,  many of the remarkable properties of the 
operators $\pi_i$ get generalized to the operators
$\pi_i(\alpha,\beta)$.

\subsection{Trivial factors}\label{sect:triv-fact}

From eq.(\ref{exch2}) we already know that as a polynomial in the 
spectral parameters, the component $\psi_\bw(\bz)$ has factors
$(z_{i+1}-\tau_{w_i})(z_i-\nu_{w_{i+1}})$ whenever $w_i<w_{i+1}$
\be\label{exch2_2}
\psi_\bw(\bz) = (z_{i+1}-\tau_{w_i})(z_i-\nu_{w_{i+1}}) \tilde \psi_\bw(\bz)
\ee
and the factor $\tilde \psi_\bw(\bz)$ is a polynomial symmetric under exchange
$z_i\leftrightarrow z_{i+1}$. 
We want to use this remark in order to determine for each component
$\psi_\bw(\bz)$  as many ``trivial'' factors of the form  $(z_i-\tau_\alpha)$ or
$(z_i-\nu_\beta)$ as possible. 
\begin{proposition}\label{prop-factors-1}
1) Suppose that given $j<k$,  the configuration $\bw$ is such that
$w_i<w_k$ for $j\leq i <k $, then 
$\psi_\bw(\bz)$ is divisible by $(z_j-\nu_{w_k})$.\\
2) Suppose that given $j<k$,  the configuration $\bw$ is such that
$w_i>w_j$ for $j<  i \leq k $, then 
$\psi_\bw(\bz)$ is divisible by $(z_k-\tau_{w_j})$.
\end{proposition}
\begin{proof}
We just prove point $1)$, point $2)$ being completely analogous. 
We proceed by a simple induction on the difference $h=k-j$. For $h=1$ the
statement follows immediately from eq.(\ref{exch2}). Now suppose the
statement true for $h$ and take $k-j=h+1$. Since $w_k>w_{k-1}$ we know
that $\psi_\bw(\bz)=\pi_{k-1}(w_k,w_{k-1})
\psi_{s_{k-1}\circ\bw}(\bz)$. By induction we know that
$\psi_{s_{k-1}\circ\bw}(\bz)$ is divisible by $(z_j-\nu_{w_k})$; on
the other hand, since $\pi_{k-1}(w_k,w_{k-1})$ acts only on the
variables $z_{k-1},z_k$ we conclude that $(z_j-\nu_{w_k})$ divides
also $\psi_\bw(\bz)$.
\end{proof}
An immediate corollary of the previous Proposition is the following
\begin{corollary}\label{zero-max-min}
1) If $w_i\neq N$ then $\psi_\bw(\bz)$ is divisible
by $(z_i-\nu_{N})$.\\
2) If $w_i\neq 1$ then $\psi_\bw(\bz)$ is divisible
by $(z_i-\tau_{1})$.
\end{corollary}
A particularly favorable situation is when the configuration $\bw$
presents a subset of consecutive increasing  entries.
For a finite set of integers $I$, a configuration $\bw$ and $1\leq
j,k\leq N$ define 
\be
G_{I;\bw;j,k}(\bz)=\prod_{i=j}^k  \left(\prod_{\substack{\alpha\in
    I\\w_i>\alpha }}(z_i-\tau_\alpha) \prod_{\substack{\alpha\in I\\w_i<\alpha }}(z_i-\nu_\alpha) \right). 
\ee
Then we have the following
\begin{proposition}\label{prop-factors-2}
Suppose that  for some $j<k$, the configuration $\bw$ is such that 
$w_i< w_\ell$ for $j\leq i <\ell\leq k$, then the component $\psi_\bw(\bz)$ is of the form
\be
\psi_\bw(\bz) = \tilde\psi_\bw(\bz)G_{\bw_{[j,k]};\bw; j,k}(\bz),
\ee
where $\tilde\psi_w(\bz)$ is a
polynomial symmetric in the variables $z_j,z_{j+1},\dots, z_k$ and
$\bw_{[j,k]}=\{w_j,w_{j+1},\dots,w_k\}$.
\end{proposition}
\begin{proof}
The presence of the factor $G_{\bw_{[j,k]};\bw; j,k}(\bz)$ is
ensured by Proposition 
\ref{prop-factors-1}.  
Since $w_i<w_{i+1}$, $G_{\bw_{[j,k]};\bw; j,k}(\bz)$ is  given
by $(z_{i+1}-\tau_{w_i})(z_i-\nu_{w_{i+1}})$ times a function symmetric
under $s_i$. Therefore from eq.(\ref{exch2_2}) we conclude that
$\tilde\psi_w(\bz)$ is symmetric under $s_i$.  
\end{proof}
Now consider the configuration $12\dots N$.
For such configuration we have   
\be
\psi_{12\dots N}(\bz) =  \tilde\psi(\bz) \prod_{i=1}^N\left(
\prod_{\alpha=1}^{i-1}(z_i-\tau_\alpha)\prod_{\alpha=i+1}^{N}(z_i-\nu_\alpha)
\right).
\ee
where $\tilde\psi(\bz)$ is symmetric in all the spectral
parameters. Any other component can be obtained from $\psi_{12\dots N}(\bz)$ 
through the action of the operators
$\pi_i(\alpha,\beta)$, which preserve symmetric factors. Therefore $\tilde\psi(\bz)$ 
appears as a factor of all
the components. This means that in the minimal degree solution of the
exchange equations, we can assume $\tilde\psi(\bz)$ to be a constant
in the spectral parameter. 
We choose the normalization of $\Psi_N(\bz)$ by setting 
\be\label{normalization}
\psi_{12\dots N}(\bz) =   \phi_N(\btau,\bnu) \prod_{i=1}^N\left(
\prod_{\alpha=1}^{i-1}(z_i-\tau_\alpha)\prod_{\alpha=i+1}^{N}(z_i-\nu_\alpha ) \right).
\ee
In particular in this way we have been able to fix the 
degree of $\Psi_N(\bz)$ as a polynomial in any of the spectral
parameters $z_i$. 
\begin{corollary}\label{coroll-degree}
The degree
of $\Psi_\bm(\bz)$  as a polynomial in $z_i$ for any $i$ is equal to
$N-1$. 
\end{corollary}

Let us discuss briefly the particular case
$\tau_\alpha=\nu_\alpha$. In such a case it is simple to see that
$$
\pi_i(\alpha,\beta)(z_i-\nu_\alpha)^{-1}(z_{i+1}-\nu_\beta)^{-1} =
(z_i-\nu_\beta)^{-1}(z_{i+1}-\nu_\alpha)^{-1}.  
$$
This implies that 
\be
\psi_{\bw}(\bz) = \prod_{1\leq i \leq L} (z_i-\nu_{w_i})^{-1}
\ee
is solution of the exchange eqs.(\ref{exch2})\footnote{In
  order to make $\psi_{\bw}(\bz)$ a polynomial is sufficient to
  multiply it by $$\prod_{1\leq i \leq N}\prod_{1\leq \alpha\leq
    N}(z_i-\nu_\alpha).$$}. This is consistent with 
\cite{karimipour1999multispecies}, indeed by setting $\bz= \infty$ we
simply obtain that the stationary measure is uniform. 

\subsection{Recursions}\label{sect:recurs}

We have already seen in Corollary \ref{zero-max-min} that if we
specialize $z_i=\tau_{1}$ or $z_i=\nu_{N}$,
then all the components whose configuration doesn't present a
particle of specie respectively $1$ or $N$ at position
$i$ are equal to zero. Here we want to characterize all the other
components under the same specialization. The first step is to compare
the specializations of $\Psi_N(\bz)$ at different positions.
Let us define the operators $\hat s_i\in \End(\bC[\bz]\otimes \mathcal
H^N)$
$$
\hat s_i \left[f(\bz)v_\bw \right] = s_i f(\bz) v_{s_i \bw}
$$
i.e. the map $\hat s_i$ transpose at the same time the variables
$z_i,z_{i+1}$ and the particles at position $i$ and $i+1$. 
There is no reason for the operators $\hat s_i$ to preserve
$\Psi_N(\bz)$ and indeed in general we have
$$
\hat s_i \Psi_N(\bz) \neq  \Psi_N(\bz),
$$
but, upon specializations $z_i=\tau_{1}$ or
$z_i=\nu_{N}$  we have the following
\begin{proposition}\label{inv-spec-prop}
\begin{align}\label{inv-spec-1}
\hat s_i \Psi_N(\bz)|_{z_i=\tau_{1}} &=
\Psi_N(\bz)|_{z_i=\tau_{1}}\\ \label{inv-spec-2}
\hat s_i \Psi_N(\bz)|_{z_i=\nu_{N}} &=  \Psi_N(\bz)|_{z_i=\nu_{N}}
\end{align}
\end{proposition}
\begin{proof}
Once written in components, eq.(\ref{inv-spec-1}) just states that
\be\label{proof-spec1}
\psi_{\bw}(\dots,z_i,z_{i+1},\dots)|_{z_i=\tau_{1}} = 
\psi_{s_i \bw}(\dots,z_{i+1},z_{i},\dots)|_{z_i=\tau_{1}}.
\ee
The previous equation is obvious for $w_i\neq 1$,  
since both side vanish.
It remains to show the case
$w_i=1$. From eq.(\ref{exch2}) it follows that $\forall \bw$,  
$\psi_\bw+\psi_{s_i \bw}$ is symmetric in $z_i\leftrightarrow z_{i+1}$
\be
\psi_\bw(z_i,z_{i+1})+\psi_{s_i \bw}(z_i,z_{i+1})=\psi_\bw(z_{i+1},z_i)+\psi_{s_i \bw}(z_{i+1},z_i)
\ee
If $w_{i}=1$ then  $w_{i+1}\neq 1$ and setting
$z_i=\tau_{1}$ we have that the terms $\psi_{s_i
  \bw}(z_i,z_{i+1})$ and $\psi_\bw(z_{i+1},z_i)$ vanish and one
remains with eq.(\ref{proof-spec1}).
The proof of eq.(\ref{inv-spec-2}) follows the same lines. 
\end{proof}
Now let us define two insertion operators on configurations $\bw$ of
length $N-1$ 
\begin{itemize}
\item $\Upsilon_j^{N}$ applied to $\bw$  inserts the entry $N$ between
  $w_{j-1}$ and $w_{j}$.
 \item $\Upsilon_j^{1}$ applied to $\bw$  inserts the entry $1$ between
  $w_{j-1}$ and $w_{j}$ and increases all the other entries by $1$.
\end{itemize}
For example 
$$
\Upsilon_3^{6}~21534 = 216534,~~~~~\Upsilon_3^{1}~21534 =321643. 
$$
and extend it to a linear map $\mathcal H_{N-1} \rightarrow \mathcal
H_{N}$, by the action on a basis
$\Upsilon_j^{\alpha}v_\bw:= v_{\Upsilon_j^{\alpha}\bw}$. Using such maps we
can  characterize completely the
specializations of $\Psi_N(\bw)$ in terms of solutions of the
exchange equations of a smaller system\footnote{Here and in the
  following, for a finite or infinite string of ordered variables like
  $\bz=\{z_1,z_2,\dots\}$, and a set if integers $I$, the notation
  $\bz_{\widehat{I}}$ means
$$
\bz_{\widehat{I}}= \bz\setminus \{z_i|i\in I\},
$$
keeping the order inherited from $\bz$.
}
\begin{theorem}[Recursion]\label{recur-theor}
Upon specializations $z_j=\tau_{1}$ or $z_j=\nu_{N}$ we
have the following identities
\begin{align}\label{recursion1}
\Psi_N(\bz)|_{z_j=\tau_{1}} &= \kappa^b_N(\bz_{\widehat j})~ \Upsilon_j^{1}
\tilde \Psi_{N-1}(\bz_{\widehat j})\\\label{recursion2}
\Psi_N(\bz)|_{z_j=\nu_{N}} &= \kappa^t_N(\bz_{\widehat j})~
\Upsilon_j^{N}\Psi_{N-1}(\bz_{\widehat j}), 
\end{align}
where $\tilde \Psi_{N-1}(\bz)$ is obtained from $\Psi_{N-1}(\bz)$ by
renaming $\tau_i,\nu_i \rightarrow \tau_{i+1},\nu_{i+1}$, and  
\begin{align}
\kappa^b_N(\bz) &=
\prod_{\alpha=2}^N(\tau_{1}-\nu_\alpha)^{\alpha-1}
\prod_{i=1}^{N-1}(z_i-\tau_{1})\\
\kappa^t_N(\bz) &=
\prod_{\alpha=1}^{N-1}(\tau_\alpha-\nu_N)^{N-\alpha}
\prod_{i=1}^{L-1}(z_i-\nu_{N}).
\end{align}
\end{theorem}
\begin{proof}
The proof of the two eqs.(\ref{recursion1},\ref{recursion2})
is completely similar, hence we prove only
eq.(\ref{recursion2}). As we already know, Corollary \ref{zero-max-min}
tells us that $\Psi_N(\bz)|_{z_j=\nu_{N}}$ is in the
image of $\Upsilon_j^{N}$, then call $\bar
\Psi^{(j)}_{N-1}(\bz_{\widehat j})$ its unique preimage. 
Actually, thanks to Proposition \ref{inv-spec-prop}, we have that 
$$\bar \Psi^{(j)}_{N-1}(\bz_{1,\dots,L-1})=\bar
\Psi^{(k)}_{N-1}(\bz_{1,\dots,L-1})$$ for $j\neq k$, hence we can
suppress the upper label $(j)$. Since  $\Upsilon_j^{N}$
intertwines the $\check R$ matrices acting on sites different from
$j-1$ and $j$
\begin{align*}
\check R_i(z,w) \Upsilon_j^{N} = \Upsilon_j^{N} \check R_i(z,w) && i<j-1 \\
\check R_{i+1}(z,w) \Upsilon_j^{N} = \Upsilon_j^{N} \check R_i(z,w) && i>j-1 
\end{align*}
it is obvious that $\bar \Psi_{N-1}(\bz_{1,\dots,L-1})$ satisfies all
the exchange equations (\ref{exchange-eq}) for $i\neq j-1$. But since $\bar
\Psi_{N-1}(\bz_{1,\dots,L-1})$ doesn't depend on $j$ it must satisfy
also the exchange equation for $i=j-1$. Hence by unicity of the
solution of the exchange equations we have
$$
\bar \Psi_{N-1}(\bz_{1,\dots,L-1}) \propto \Psi_{N-1}(\bz_{1,\dots,L-1}).
$$
The proportionality factor is fixed by looking at the
specialization of the component associated to the configuration
$12\dots N$. 
\end{proof}
While the recursion relations of Theorem \ref{recur-theor} do not
allow to reconstruct recursively the full vector $\Psi_N(\bz)$ (recall
that $\deg_{z_j}\Psi_N(\bz)=N-1$), it can be used to determine the
form of certain families of components. For $1\leq \beta\leq N$ define
$$
\bw^{(\beta,N)}=12\dots \widehat \beta\dots N\beta
$$
Thanks to Propositions \ref{prop-factors-1} and \ref{prop-factors-2}
we know that the component $\psi_{\bw^{(\beta,N)}}(\bz)$ has the following form
\be\label{first-non-trivial}
\psi_{\bw^{(\beta,N)}}(\bz)= \phi^{(\beta)}_N(\bz)\widetilde
\psi_{\bw^{(\beta,N)}}(\bz_{\widehat N})
\ee
where 
$
\phi^{(\beta)}_N(\bz)= G_{[1,\dots,N];\bw^{(\beta,N)};N,N}(\bz)G_{[1\dots \widehat \beta
    \dots N];\bw^{(\beta,N)};1,N-1}(\bz)
$
and $\widetilde \psi_{\bw^{(\beta,N)}}(\bz_{\widehat N})$ is a symmetric
polynomial of degree $1$ in each of its variables. 
In this case the recursion relations
(\ref{recursion1},\ref{recursion2}), are enough to completely fix
$\widetilde \psi_{\bw^{(\beta,N)}}(\bz_{\widehat N})$ in a recursive
way. For this we 
introduce the following family of polynomials in the variables
$\bz,{\bf t},{\bf v}$, indexed by two non negative integers 
$r$ and $s$ 
\be
\mathfrak S^{r, s}(\bz,{\bf t},{\bf v})=\prod_{\substack{1\leq \alpha \leq
  r\\ 1\leq \beta \leq s}}(t_\alpha-v_\beta)
\oint_{\bf t} \frac{dw}{2\pi i}
\frac{\prod_{j=1}^{r+s-2}(z_i-w)}{\prod_{\alpha=1}^r(w-t_\alpha)\prod_{\beta=1}^s(w-v_\beta)
}, 
\ee
where the contour integration encircles only the poles at ${\bf t}$. 
It is easy to see that the polynomials $\mathfrak S^{r, s}(\bz,{\bf
  t},{\bf v})$ are fully characterized (by Lagrange interpolation) by
the following recursion relations  
\begin{align}\label{mathfrakS1}
\mathfrak S^{r,s}(\bz,{\bf t},{\bf
  v})|_{z_j=t_{\alpha}}=-\prod_{1\leq \beta\leq
    s}(t_\alpha-v_\beta) \mathfrak S^{r-1, s}(\bz_{\widehat j},{\bf
    t}_{\widehat \alpha},{\bf v}),\\\label{mathfrakS2}
\mathfrak S^{r,s}(\bz,{\bf t},{\bf
  v})|_{z_j=v_{\beta}}=-\prod_{1\leq \alpha \leq
  r}(t_\alpha-v_\beta) \mathfrak S^{r,s-1}(\bz_{\widehat j},{\bf
  t},{\bf v}_{\widehat \beta}),
\end{align}
for $1\leq \alpha \leq r, 1\leq \beta \leq s$, $1\leq j\leq r+ s-2$ and boundary conditions
$$
\mathfrak S^{0, s}(\bz,{\bf t},{\bf
  v})=\mathfrak S^{r,0}(\bz,{\bf t},{\bf
  v})=0.
$$
\begin{proposition}
\begin{multline}\label{sol-recursion}
\widetilde \psi_{\bw^{(\beta,N)}}(\bz)=
\prod_{\substack{1\leq \alpha<\gamma \leq \beta\\
\&\\
\beta\leq \alpha<\gamma \leq N}}(\tau_\alpha-\nu_\gamma)^{\gamma-\alpha-1}
\prod_{1\leq \alpha <\beta <\gamma\leq N}(\tau_\alpha-\nu_\gamma)^{\gamma-\alpha-2}\\
\mathfrak S^{\beta,
  N-\beta+1}(\bz,\{\tau_1,\tau_2\dots,\tau_\beta\},\{\nu_\beta,
\nu_{\beta +1},\dots,\nu_N\}) 
\end{multline}
\end{proposition}
\begin{proof}
As mentioned before, being $\widetilde \psi_{\bw^{(\beta,N)}}(\bz)$ a symmetric
polynomial of degree $1$ in each variable, it is completely determined
by the recursion relations
(\ref{recursion1},\ref{recursion2}). Therefore it is sufficient to
check that plugging eq.(\ref{sol-recursion}) into
eq.(\ref{first-non-trivial}), we obtain a family of polynomials that satisfy the
recursion relations. This is readily done using the recursions for
$\mathfrak S^{r,s}(\bz,{\bf t},{\bf v})$,
eqs.(\ref{mathfrakS1},\ref{mathfrakS2}) 
\end{proof}
At this point we remark the appearance of double Schubert polynomials.
Indeed, let $\sigma(h,N)\in \mathcal S^N$ be the permutation defined by 
$$
\sigma(\beta,N) =(1,\beta+1,\beta+2,\dots,N,2,3,\dots,\beta).
$$
as will be shown in Appendix \ref{app-schubert}, the polynomial $\mathfrak
S^{\beta, N-\beta+1}(\infty,{\bf t},{\bf   v})$ is the double Schubert
polynomial in the variables ${\bf t},{\bf v}$ associated to the
permutation $\sigma(N-\beta+1,N)$
\be
\mathfrak S^{\beta , N-\beta+1}(\infty,{\bf t},{\bf v})= \mathfrak
S_{\sigma(N-\beta+1,N)}({\bf t},{\bf v}). 
\ee

\subsection{Descending configurations}\label{sect:desc-conf}

Let $\bw(N,h)$ be the configuration given by 
$$
\bw(N,h)= h(h-1)\dots 21(h+1)(h+2)\dots N
$$
As particular cases we have that 
\begin{align*}
\bw(N,1)=12\dots(N-1)N,&&~~~~\bw(N,N)=N(N-1)\dots 2 1&&
\end{align*}
One of the key results of this paper is the formula for
$\psi_{\bw(N,h)}(\bz)$, which is the content of the following
\begin{theorem}\label{theo-product-general}
Let $1\leq h \leq N$, the polynomial $ \psi_{\bw(N,h)}(\bz)$ has the
following form
\begin{multline}\label{product-general}
\psi_{\bw(N,h)}(\bz) =  \phi^{(N,h)}(\btau,\bnu)\\G_{[h+1,N],\bw(N,h),1,L}(\bz)
\prod_{\beta=1}^h \mathfrak S^{(\beta,N-\beta+1)}(\bz_{\widehat {h-\beta+1}},\btau,\bnu^c),
\end{multline}
where
$$
\phi^{(N,h)}=\prod_{1\leq \alpha\leq h<\gamma \leq N}
(\tau_\alpha-\nu_\gamma)^{\gamma-h-1}
\prod_{h<   \alpha <\gamma\leq N}
(\tau_\alpha-\nu_\gamma)^{\gamma-\alpha-1},
$$
where for a fixed $N$, $\bnu^c$ correspond to $\bnu$ taken in
reversed order starting from $N$, i.e.
\be
\nu^c_i= \nu_{N-i+1}.
\ee
\end{theorem}
For later reference let us write explicitly the case $h=1$, which
provides a nice factorized expression for $\psi_{N(N-1)\dots1}(\bz) $ 
\be\label{eq-largest}
\psi_{N(N-1)\dots1}(\bz) = \prod_{\beta=1}^N \mathfrak
S^{(\beta,N-\beta+1)}(\bz_{\widehat {N-\beta+1}},\btau,\bnu^c). 
\ee
Let us also notice that
as a corollary of Theorem \ref{theo-product-general}, by setting $\bz= {\bf \infty}$ we prove a
generalization of Lam-Williams Conjectures $3$ and $4$ \cite{lam2012markov}
\begin{corollary}
With the normalization eq.(\ref{norm-general}) the component reads
$$
\psi_{\bw(N,h)}= \phi^{(N,h)}(\btau,\bnu) \prod_{\beta}^h\mathfrak
S_{\sigma(\beta,N)}(\btau,-\bnu^c). 
$$
\end{corollary} 
For the proof of Theorem \ref{theo-product-general} we need some
preparatory steps. Let us introduce the families $\mathcal D(N,h)$, which consist
of configurations containing a 
sub-string of consecutive entries of the form 
$$
h(h-1)\dots 1.
$$
Notice that we have $\bw(N,h)\in\mathcal D(N,h)$ and $\mathcal D(N,h)\subset
\mathcal D(N,k)$ for $k\leq h$. Any configuration $\bw\in \mathcal D(N,h)$ can
be obtained from 
any other configuration 
$\widetilde \bw\in \mathcal D(N,h)$ by a sequence of permutations $S_1,\dots,S_\ell$
\be\label{bw_tilde_bw}
\bw = S_\ell\dots S_2S_1\widetilde \bw
\ee
with $S_j$ being either a transposition of consecutive
entries $h<w_{i+1}<w_i$, $s_i:w_iw_{i+1}\mapsto w_{i+1}w_i
$ or a permutation $\sigma_i$ moving $w_i$ from the left of the string
$h(h-1)\dots 1$ to its right, i.e. 
$$
\sigma_i:\bw_L w_ih(h-1)\dots 1\bw_R \mapsto \bw_L h(h-1)\dots
1w_i\bw_R .
$$
\begin{proposition}\label{prop-factors-desc}
Let $\widetilde \bw\in \mathcal D(N,h)$, where the sub-string
$h(h-1)\dots 1$ goes from $\tilde j+1$ to $\tilde j+h$. Suppose 
 $\psi_{\widetilde  \bw}(\bz)$ to be of the form 
$$
\psi_{\widetilde
  \bw}(\bz) = \psi^{(0)}_{\widetilde
  \bw}(\bz) \bar\psi_{N;\tilde j+1,\tilde j+h}(\bz)
$$
with $ \psi^{(0)}_{\widetilde
  \bw}(\bz)$  a polynomial symmetric in the variables
$\bz_{\widehat{[\tilde j+1,\dots,\tilde j+h]})}$ and 
$$
\bar\psi_{N;\tilde j+1,\tilde j+h}(\bz):=\prod_{\beta=1}^h \mathfrak
S^{(\beta,N-\beta+1)}(\bz_{\widehat {\tilde j+\beta}},\btau,-\bnu^c).
$$
Then the same factorization holds for the components associated to
any $\bw\in \mathcal D(N,h)$  with sub-string $h(h-1)\dots 1$
between positions $j+1$ and $j+h$, i.e. the component
$\psi_\bw(\bz)$ is also of the form
\be\label{eq-bw_desc}
\psi_\bw(\bz)=\psi^{(0)}_{\bw}(\bz) \bar\psi_{N;j+1,j+h}(\bz), 
\ee
for some $\psi^{(0)}_{\bw}(\bz)$, symmetric in $\bz_{\widehat{[j+1,j+h]}}$.
\end{proposition}
\begin{proof}
Call $d(\widetilde \bw, \bw)$ the minimal $\ell$ for which we
have a sequence of permutations realizing eq.(\ref{bw_tilde_bw}).
We prove the statement by induction on 
$d(\widetilde \bw, \bw)$. For $d(\widetilde \bw,
\bw)=0$ there is nothing to prove. Now consider $\ell=d(\widetilde \bw,
\bw)>0$, this means that we can write $\bw=S_\ell\dots S_1\bw$, for
some $S_i$ as defined above. By induction we know that for
$\bw'=S_{\ell-1}\dots S_1\bw$, $\psi_{\bw'}(\bz)$ is of the form
eq.(\ref{eq-bw_desc}). For $S_\ell$ there are two possibilities
either $S_\ell=s_i$ or $S_\ell=\sigma_i$ for some $s_i$ or $\sigma_i$.
 
If $S_\ell=s_i$, it means that $\psi_{\bw}(\bz) =
\pi_i(w'_i,w'_{i+1})\psi_{\bw'}(\bz)$.
The polynomial $\psi_{\bw'}(\bz)$ is
of the form eq.(\ref{eq-bw_desc}) with the factor
$\bar\psi_{\bm,\rho,\bw}(\bz)$ symmetric in $z_i,z_{i+1}$, therefore
we have $\psi_{\bw}(\bz) =
\left[\pi_i(w'_i,w'_{i+1})\psi^{(0)}_{\bw}(\bz)\right]\bar\psi_{N,j'+1,j'+h}(\bz)$,
which is again of the form eq.(\ref{eq-bw_desc}).

If $S_\ell=\sigma_i$ then 
$
\psi_{\bw}(\bz) =\pi_{\sigma_i}(\bw') \psi_{\bw'}(\bz). 
$
Since $\psi_{\bw'}(\bz)$ is of the form eq.(\ref{eq-bw_desc}), we can
use Proposition \ref{general-exch-prop}\footnote{Up to a translation
  in the indices of the variables $z_k\rightarrow z_{k-i+1}$.}, with $u=w_i$,
\begin{align*}
K(z_i;\bz_{i+1,\dots,i+h})= \psi^{(0)}_{\bw}(\bz),&&
F_j(w)= \frac{\prod_{\substack{1\leq \ell\leq N\\\ell \notin [i+1,i+h]
  }} (z_\ell-w)}{\prod_{\beta=h}^N(w-\nu_\beta)}.&& 
\end{align*}  
Then eq.(\ref{eq-general-exch-prop}) allows to conclude that
$\psi_{\bw}(\bz)$ is of the desired form.
\end{proof}
We can now pass to the proof of Theorem
\ref{theo-product-general}.
\begin{proof}[Proof. of Theorem \ref{theo-product-general}] 
The proof proceeds by a double induction on $N$ and on $h$. 

For $N=1$ the statement is trivial, while 
the case $h=1$  
holds because in such case $\bw(N,1)=12\dots N$ and
eq.(\ref{product-general}) coincides with eq.(\ref{normalization}).

Assuming $N,h>1$ we proceed by factor exhaustion. The
presence of the factor $G_{[h+1,N],\bw(N,h),1,L}(\bz)$ is
ensured by Propositions \ref{prop-factors-1} and
\ref{prop-factors-2}.  By induction we know that
$\psi_{\bw(N,h-1)}(\bz)$ is of the form eq.(\ref{product-general}) and in
particular of the form eq.(\ref{eq-bw_desc}), therefore it follows
from Proposition \ref{prop-factors-desc} that the components
corresponding to configurations in $\mathcal D(N,h-1)$ have the same
form. Since $\bw(N,h)\in \mathcal D(N,h-1)$, we conclude in particular
that $\psi_{\bw(N,h)}(\bz)$ 
contains the factor
$$\prod_{\beta=1}^{h-1} \mathfrak S^{(\beta,N-\beta+1)}(\bz_{\widehat
  {h-\beta+1}},\btau,\bnu^c),$$ 
which is prime with  $G_{[h+1,N],\bw(N,h),1,L}(\bz)$.
The remaining factor $g(\bz)$ is a symmetric polynomial in the
variables  $z_h,z_{h+1},\dots,z_N$, of degree $1$ in
each of these variables.
Therefore, in order to check that 
$$
g(\bz)= \phi^{(N,h)}(\btau,\bnu) \mathfrak S^{(h,N-h+1)}(\bz_{\widehat
  1},\btau,\bnu^C) 
$$ 
it is enough to check that eq.(\ref{product-general}) holds when
specialized at two distinct values of $z_i$  for $i\in [h,N]$. 
For $i=N$, using the recursions of Theorem
\ref{recur-theor}, the specialization
$\psi_{\bw(N,h)}(\bz)|_{z_N=\nu_{N}}$ can be written in terms of 
$\psi_{\bw(N-1,h)}(\bz_{\widehat N})$, which by induction
(it corresponds to the case $N-1$) is given by
the expression of eq.(\ref{product-general}). It is not difficult to
check that using eq.(\ref{product-general}) for
$\psi_{\bw(N,h)}(\bz)$ and specializing $z_N=\nu_{N}$ one
obtains the same result. In the same way one can check the specialization
$z_h=\tau_{1}$.
\end{proof}
Another result that can be obtained by using Theorem 
\ref{theo-product-general} concerns the primality of the components of $\Psi_N(\bz)$
\begin{theorem}\label{primality-general}
With the normalization given by eq.(\ref{normalization}) the
components of $\Psi_N(\bz)$, as 
functions of $\btau,\bnu$ and $\bz$, are prime polynomials with
integer coefficients.
\end{theorem}
\begin{proof}
First we notice that if $F(\btau,\bnu;\bz)$ is a polynomial in
$\btau,\bnu,\bz$ with integer coefficients and
$\pi_i(\beta,\alpha)F(\btau,\bnu;\bz)$ is also polynomial in
$\btau,\bnu,\bz$, then it must have integer
coefficients as well. Therefore, since all the components can be obtained from
$\psi_{12\dots N}(\bz)$ (which has integer coefficients) by action of operators
$\pi_i(\beta,\alpha)$, once we will have proven that all components are
polynomial in $\btau,\bnu,\bz$ we shall automatically get that their
components are integer.
 
From their formulas, we see that
$\psi_{12\dots N}(\bz)$ and $\psi_{N(N-1)\dots 1}(\bz)$ are prime
polynomials in all the variables $\bz, \btau,\bnu$ (they have no
common polynomial factor), therefore the only thing that remains to be
proven is the polynomiality in $\btau$ and $\bnu$ of all the
other components.
Any $\psi_\bw(\bz)$ can be obtained from $\psi_{N(N-1)\dots 1}(\bz)$,
by sequential action of the operators $\pi_\ell(\alpha,\beta)$ with
$1\leq \ell \leq L-1$. In particular if for $i<j$,  $w_i<w_j$, then
exactly one of the members of the sequence of operators
$\pi_\ell(\alpha,\beta)$ is of the form $\pi_\ell(w_i,w_j)$ for some
$\ell$. Therefore, in principle we could get a factor at
denominator of the form $\tau_{w_i}-\nu_{w_j}$, screwing up the
polynomiality in $\btau$ and $\bnu$.
We have to make sure this does not happen. 
In facts the reasoning above can be reversed giving us some positive
information, namely it tells us that if for $j<i$,
$w_i<w_j$, then at denominator we do not have a factor of the form
$\tau_{w_i}-\nu_{w_j}$, because there is no way this could have been
arisen. 
Now let us come back to the case  $i<j$,  $w_i<w_j$: upon rotating by
$h= L+1-j$ steps we get a new configuration $\tilde \bw=R^h\bw$ with 
$\tilde w_1= w_j$ and $\tilde w_{L+1+i-j}= w_i$. Since $L+1+i-j>1$ we
conclude that 
$\psi_{R^h\bw}(\bz)$ (and henceforth $\psi_{\bw}(\bz)$) does not
have the factor $\tau_{w_i}-\nu_{w_j}$ at denominator.
\end{proof}
A corollary of the previous result is Theorem \ref{theo-minimal-poly}.

\subsection{Factorization and normalization}\label{sect:normal}

Let us now draw a few consequences of Theorem \ref{theo-product-general}.
\begin{theorem}
Let $1\leq h \leq  k \leq N$ and $\bw$ of the form
$$
\bw= \bw^{(L)}k(k-1)\dots(h+1)h \bw^{(R)}
$$
with $w^{(L)}_i>k$ and $w^{(r)}_i<h$, then
$\psi_{\bw^{(L)}k(k-1)\dots(h+1)h \bw^{(R)}}(\bz)$ has a factor of the
form
\be
\prod_{\beta=h}^k \mathfrak
S^{(\beta,N-\beta+1)}(\bz_{\widehat {N-\beta+1}},\btau,\bnu^c)
\ee
\end{theorem}
\begin{proof}
The configuration $\bw= \bw^{(L)}k(k-1)\dots(h+1)h \bw^{(R)}$ is
obtained from the descending configurations $N(N-1)\dots 1$ through
transpositions involving only the first $N-k$ and the last $h-1$
positions. Therefore $\psi_{\bw^{(L)}k(k-1)\dots(h+1)h
  \bw^{(R)}}(\bz)$ is obtained from $\psi_{N(N-1)\dots 1}(\bz)$ by the 
action of operators $\pi_i(\alpha,\beta)$ with $1\leq i \leq N-k-1$ or
$N-h+1\leq i\leq N-1$. Any factor of $\psi_{N(N-1)\dots 1}(\bz)$
which is symmetric in the first $N-k$ and the last $h-1$ variables is
preserved by the action of such operators and is a factor of $\psi_{\bw^{(L)}k(k-1)\dots(h+1)h
  \bw^{(R)}}(\bz)$.  The product $\prod_{\beta=h}^k \mathfrak
S^{(\beta,N-\beta+1)}(\bz_{\widehat {N-\beta+1}},\btau,\bnu^c)$ is
such a factor.
\end{proof}
In the same spirit we have the following factorization result
\begin{theorem}\label{factorization-probas-theo}
Suppose that $\bw = {\bf x}{\bf y}$ with $x_i>y_j$, with $\ell({\bf x})=k$, then 
\be\label{factorization-probas}
\psi_{{\bf x}{\bf y}}(\bz)= \psi^{(1)}_{{\bf x}}(\bz)\psi^{(2)}_{{\bf
    y}}(\bz), 
\ee
where $\psi^{(1)}_{{\bf x}}(\bz)$ is symmetric in $z_{k+1},\dots,z_N$ and
depends only on ${\bf x}$,
while $\psi^{(2)}_{{\bf y}}(\bz)$ is symmetric in $z_1,\dots,z_k$ and
depends only on ${\bf y}$.
\end{theorem}
\begin{proof}
Using eq.(\ref{eq-largest}), we write $\psi_{N(N-1)\dots 2 1}(\bz)$ 
as   
$$
\psi_{N(N-1)\dots 2 1}(\bz) = \psi^{(1)}_{N,k}(\bz)\psi^{(2)}_{N,k}(\bz)
$$ 
with
$$
\psi^{(1)}_{N,k}(\bz) =\prod_{\beta=N-k+1}^N \mathfrak
S^{(\beta,N-\beta+1)}(\bz_{\widehat {N-\beta+1}},\btau,-\bnu^c)  
$$
$$
\psi^{(2)}_{N,k}(\bz) = \prod_{\beta=1}^{N-k} \mathfrak
S^{(\beta,N-\beta+1)}(\bz_{\widehat {N-\beta+1}},\btau,-\bnu^c). 
$$ 
Notice that $\psi^{(1)}_{N,k}(\bz) $ is symmetric in the last $N-k$
variables, while $\psi^{(2)}_{N,k}(\bz)$ is symmetric in the first $k$
variables. 

Any configurations $\bw = {\bf x}{\bf y}$ with $x_i>y_j$, with
$\ell({\bf x})=k$ is obtained from $N(N-1)\dots 1$ by separate
transpositions of the first $k$ and last $N-k$ and therefore
$\psi_{{\bf x}{\bf y}}(\bz)$ is obtained from $ \psi_{N(N-1)\dots 2
  1}(\bz)$ by action of operators $\pi_i(\alpha,\beta)$ with 
$i\neq k$. Being $\psi^{(2)}_{N,k}(\bz)$ symmetric in the first $k$
variables, the operators $\pi_i(\alpha,\beta)$ with $i<k$ act only on
the first factor $\psi^{(1)}_{N,k}(\bz)$, viceversa    the operators
$\pi_i(\alpha,\beta)$ with $i>k$ act only on the second factor
$\psi^{(2)}_{N,k}(\bz)$ and therefore we get the factorization of eq.(\ref{eq-largest}).
\end{proof}
Theorem \ref{factorization-probas-theo} means in particular that for
the stationary measure, under the 
conditioning that the configuration $\bw$ splits as $\bw = {\bf
  x}{\bf y}$, with $x_i>y_j$, then ${\bf x}$ and ${\bf y}$ are
independent. This fact has already been proven for the case
$\bz =\infty$ and $\nu_\alpha=0$ in \cite{aas2013product} using the
multiline queues representation of $\psi_\bw$.

The next result concerns the \emph{partition function}
$\cZ_N(\bz,\btau,\bnu)$, that is the sum of all the components
\be
\cZ_N(\bz,\btau,\bnu) := \sum_{\ell (\bw)=N} \psi_\bw(\bz).
\ee
In order to write the formula for $\cZ_N(\bz,\btau,\bnu)$ we need a
little bit of further notation. For an ordered set of $N$ variables
$\bz=\{z_1,\dots,z_N\}$ and a permutation $\sigma \in \mathcal S_N$ we
write
$
\bz_\sigma:= \{z_{\sigma(1)},\dots, z_{\sigma(N)}\}. 
$
For any function $f(\bz)$ of $N$ variables, we call $\textrm{Sym}\left[f(\bz)
  \right]$ its symmetrized version, i.e.
$$
\textrm{Sym}\left[f(\bz) \right] = \sum_{\sigma\in\mathcal S_N}f(\bz_\sigma).
$$
We have the following
\begin{theorem}[Partition function]\label{part-funct-theo}
$$
\cZ_N(\bz,\btau,\bnu)= \frac{1}{\prod_{1\leq \alpha<\beta\leq
    N}(\tau_\alpha-\nu_\beta)^{\beta-\alpha}}\textrm{Sym}\left[ \frac{ 
    \psi_{1\cdots(N-1)N}(\bz_{w_0})\psi_{N(N-1)\cdots 1}(\bz) }{\Delta(\bz)} \right] 
$$
Where $w_0$ is the longest permutation in $\mathcal S_N$
$w_0=(N,N-1,\dots,2,1).$
\end{theorem}
\begin{proof}
As we noticed before, any component can be obtained from $\psi_{N(N-1)\cdots 1}(\bz)$
by sequential action of operators $\pi_i(\alpha,\beta)$, with $1\leq
i\leq N-1$. By expanding the divided difference operator we get
\be\label{perm_exp-groth}
\psi_{\bw}(\bz) =  \sum_{\sigma \in \mathcal
    S_N} k_{\bw,\sigma}(\bz)~ \psi_{N(N-1)\cdots 1}(\bz_\sigma)
\ee
for some coefficients $ k_{\bw,\sigma}(\bz)$. It is not difficult to
see that if $\ell(\sigma)$ is larger than the length of the shortest
permutation mapping $N(N-1)\cdots 1$ to $\bw$ then 
$
k_{\bw,\sigma}(\bz)=0.
$
In particular if $\sigma$ is the longest permutation $w_0$ in
$\mathcal S_N$ then 
$ k_{\bw,w_0}(\bz)\neq 0 $ only for $\bw=12\cdots N$.
In such a case, by explicitly expanding the divided difference operators,
we easily find
$$
k_{12\cdots N,w_0}(\bz) =
\frac{(-1)^{\frac{N(N-1)}{2}}\psi_{1\cdots(N-1)N}(\bz)}{\left(\prod_{1\leq \alpha<\beta\leq
    N}(\tau_\alpha-\nu_\beta)^{\beta-\alpha}  \right)\Delta(\bz)}.
$$
By using eq.(\ref{perm_exp-groth}) we can write $\cZ(\bz,\btau,\bnu)$ as  
$$
\cZ(\bz,\btau,\bnu) = \sum_{\sigma \in \mathcal
    S_N} K_\sigma(\bz)~ \psi_{N(N-1)\cdots 1}(\bz_\sigma).  
$$
where $K_\sigma(\bz)= \sum_\bw k_{\bw,\sigma}(\bz)$. 
In the case  $\sigma=w_0$,
thanks to the discussion above, the sum giving  $K_{w_0}(\bz)$ reduces
to a single term
$$
K_{w_0}(\bz) = \sum_\bw k_{\bw,w_0}(\bz) = k_{12\cdots N,w_0}(\bz).
$$
On the other hand we know that $\cZ(\bz,\btau,\bnu)$ is symmetric in
$\bz$, therefore we have 
$$
K_\sigma(\bz)=  K_e(\bz_\sigma),
$$
which concludes the proof.
\end{proof}
Since both $\psi_{1\cdots(N-1)N}(\bz_{w_0})$ and $\psi_{N(N-1)\cdots
  1}(\bz)$ have factorized expression, the formula obtained in Theorem
\ref{part-funct-theo} can be recast in a determinantal form
\be
\cZ(\bz,\btau,\bnu)=\frac{1}{\Delta(\bz)\prod_{1\leq \alpha<\beta\leq
    N}(\tau_\alpha-\nu_\beta)}~ \det_{1\leq \alpha,\beta
\leq N} M^{(N)}_{\alpha,\beta}
\ee
where 
$$
M^{(N)}_{\alpha,\beta} = \mathfrak
S^{(\beta,N-\beta+1)}(\bz_{\widehat \alpha},\btau,\bnu^c) \prod_{1\leq
  \gamma < N-\beta+1}(z_\alpha-\tau_\gamma)\prod_{N-\beta+1<\gamma\leq N}(z_\alpha-\nu_\gamma).
$$

\appendix

\section{Schubert polynomials}\label{app-schubert}

The Double Schubert polynomials play an important role in the geometry of flag
varieties, where they represent equivariant cohomology classes, and in
the combinatorics of the Bruhat order of the symmetric group (see
\cite{lascoux1985schubert, macdonald1991notes, manivel2001symmetric}).  
In this Appendix we recall their definition and provide an explicit
formula for a certain class of permutations.

Let $\bt=\{t_1,t_2,\dots \}$ and
$\bv=\{v_1, v_2,\dots\}$ be two infinite sets of variables.
In this section we use the divided difference operators in the
variables $\bt$, so here $s_i$ is the transposition of the variables
$t_i \leftrightarrow t_{i+1}$
$$
\partial_i = \frac{1-s_i}{t_i-t_{i+1}}.
$$ 
These operators satisfy the following relations
\be
\begin{split}
\partial_i^2&=0\\
\partial_i\partial_j&=\partial_j\partial_i \hspace{1cm}|i-j|>1\\
\partial_i\partial_{i+1}\partial_i&=\partial_{i+1}\partial_i\partial_{i+1}
\end{split}
\ee
Let $\mathcal S^{\infty}$ be the infinite symmetric group, the algebra
generated by $\partial_i$ for $i\geq 1$ has a basis indexed by 
permutations $\sigma \in \mathcal S^{\infty}$. Let $s_{i_\ell}\cdots
s_{i_1}$ be a reduced decomposition of $\sigma$, then  
$$
\partial_\sigma:= \partial_{i_\ell}\cdots\partial_{i_1}
$$
is well defined, i.e. it is the same for different reduced
decomposition of the same permutation. 

\begin{definition}[Double Schubert polynomials]
The double Schubert polynomials are a family of polynomials $\mathfrak
S_{\sigma}(\bt,\bv)$ in the variables $\bt,\bv$ indexed by 
permutations $\sigma \in \mathcal S^{\infty}$.
For $\sigma \in \mathcal S^{N}\subset \mathcal S^{\infty} $,
$\mathfrak S_{\sigma}(\bt,\bv)$ is defined by
\begin{equation}\label{def-dsp}
\mathfrak S_{\sigma}(\bt,\bv) = \partial_{\sigma^{-1} w_{N}}\prod_{i+j\leq N}(t_i-v_j)
\end{equation}
where $w_{N}$ is the longest permutation in $\mathcal
S^{N}\subset \mathcal S^{\infty}$.
\end{definition}
Now let $\sigma(h,N)\in \mathcal S^{N}$ defined by
$$\sigma(h,N)=(1,h+1,h+2,\dots N,2,3,\dots,h).$$
We show that 
\be\label{sigma(h,N)}
\mathfrak S_{\sigma(h,N)}(\bt,\bv) = \oint_\bt
\frac{dw}{2\pi i} \frac{\prod_{\substack{1\leq \alpha
    \leq N-h+1\\1\leq \beta \leq h}}(t_\alpha-v_\beta)}{\prod_{ 1\leq \alpha
    \leq N-h+1} (w-t_\alpha)\prod_{1\leq \beta \leq h}(w-v_\beta)}
\ee
First notice that for $\tilde \sigma(h,N)=(h+1,h+2,\dots
  N,1,2,3,\dots,h)$ we have  
$$
\mathfrak S_{\sigma'(h,N)}(\bt,\bv) =\prod_{\substack{1\leq \alpha
    \leq N-h\\1\leq \beta \leq h}}(t_\alpha-v_\beta).
$$
Then, using the definition eq.(\ref{def-dsp}),  we can write
$$
\mathfrak S_{\sigma(h,N)}(\bt,\bv)
= \partial_1\cdots \partial_{N-h-1}\partial_{N-h}\mathfrak
S_{\sigma'(h,N)}(\bt,\bv) 
$$
and eq.(\ref{sigma(h,N)}) follows from this general Lemma 
\begin{lemma}\label{multi-exch-lemma0}
Let $K(t_1,t_2,\dots,t_k;t_{k+1})$ be a symmetric function in the
variables $\{t_1,t_2,\dots,t_{k}\}$, then the following identity holds
\begin{multline}
  \partial_1\cdots \partial_{k-1}\partial_k K(t_1,t_2,\dots,t_k;t_{k+1})
  =\\ \sum_{i=1}^{k+1} \frac{K(t_1,\dots, \widehat {t_i},\dots,
    t_{k+1};t_i)}{\prod_{1\leq j\neq i \leq k+1}(t_i-t_j)} 
\end{multline} 
\end{lemma}
\begin{proof}
If we act with $\partial_j$ with $1\leq j<k$, on the l.h.s. and use
the braiding 
relations of the operators $\partial_i$ we get zero, therefore the
l.h.s. is symmetric in the variables $t_1,\dots,t_k,t_{k+1}$. On the
other hand, by developing the action of the divided difference operators,  we know
that the l.h.s. can be written as  
$$
\sum_{i=1}^{k+1} K(t_1,\dots, \widehat {t_i},\dots, t_{k+1};t_i) G_i({\bf t}).
$$
Therefore it is enough to compute one of the coefficients 
$G_i({\bf  t})$. 
The term $G_{1}({\bf x})K(t_2,\dots,t_k,t_{k+1};t_1)$ in the previous
equation can be obtained in a unique way from the expansion of the
divided difference operators, namely it is given by 
$$\frac{1}{x_{1}-x_{2}}s_1\cdots\frac{1}{t_{k-1}-t_k}s_{k-1}
\frac{1}{t_k-t_{k+1}}s_kK(t_1,\dots,t_k;t_{k+1})$$ 
and hence we have
$
G_{1}({\bf t})= \frac{1}{\prod_{j=2}^{k+1}(t_{1}-t_j)}
$.
\end{proof}

\section{Technical results}\label{app-proofs}

In this Appendix we are going to present some technical results needed
for the proof of Theorem \ref{theo-product-general}. 
Let  $\bz=\{z_0,z_1,\dots,z_n\}$ and $m\leq n$, then
define the following functions 
\be
G^{(F,m)}(\bz):= \oint_{\bf t} \frac{dw}{2\pi
  i}\frac{\prod_{j=1}^n(z_j-w)F(w)}{\prod_{j=1}^m(w-\tau_j)}
\ee
\begin{proposition}\label{prop-sum-1}
The following identity holds
\begin{equation}
\sum_{j=0}^n \frac{G^{(F,m)}(\bz_{\widehat
    j})\prod_{\alpha=1}^m(z_j-\tau_\alpha)}{\prod_{0\leq i\neq j \leq n}(z_j-z_i)} =0.
\end{equation}
\end{proposition}
\begin{proof}
Consider the function $\tilde G(y):=G^{\left ( \frac{F}{y-w}
  ,m \right)}(\bz)\prod_{\alpha=1}^m(y-\tau_\alpha)$. It is simple to
see that $\tilde G(y)$ is polynomial in $y$ of degree strictly less than $m$. 
Therefore the contour integral
$$
\sum_{j=0}^n\frac{\tilde G(z_j)}{\prod_{0\leq i\neq j\leq n}(z_j-z_i)}= \oint_{\bz}\frac{dy}{2\pi i}
\frac{\tilde G(y)}{\prod_{j=0}^n(y-z_j)} =0.
$$
Then in order to conclude it is sufficient to notice that 
$$ G^{\left ( \frac{F}{y-w}
  ,m\right)}(\bz){\Large |}_{y=z_j}= G^{(F,m)}(\bz_{\widehat j}).$$
\end{proof}
\begin{proposition}\label{general-exch-prop}
Let $\bz= \{z_1,\dots ,z_{h+1}\}$, take $K(z_1;\bz_{\widehat 1})$ to
be a symmetric function in the variables 
$\bz_{\widehat 1}$ then  the following identity holds
\begin{equation}\label{eq-general-exch-prop}
\begin{split}
\pi_h(u;1)\pi_{h-1}(u;2)\cdots \pi_1(u;h)\left(K(z_1;\bz_{\widehat 1}) \prod_{j=1}^{h}
G^{(F_j,j)}(\bz_{\widehat {h-j+2}})\right)
=\\ \sum_{j=1}^{h+1}\frac{K(z_j;\bz_{\widehat
    j})\prod_{i=1}^h(z_j-\tau_i)}{\prod_{1\leq i\neq j\leq
    h+1}(z_j-z_i)} 
\prod_{j=1}^{h}\left(\frac{(z_j-\nu_u)}{(\tau_j-\nu_u)}G^{(F_j,j)}(\bz_{\widehat{h-j+1}})
\right). 
\end{split}
\end{equation}
\end{proposition}
\begin{proof}
We prove the statement by induction on $h$. For $h=1$
the statement is immediate to check\footnote{Notice that $$G^{(F,1)}(\bz)=
F(\tau_1)\prod_j(z_j-\tau_1).$$}.

Now assume $h>1$.
We start by applying $\pi_1(u, h)$ on $H(\bz)$. Since
the product of the first $h-1$ terms,  $\prod_{j=1}^{h-1}
G^{(F_j,j)}(\bz_{\widehat {h-j+2}})$, is 
symmetric in the variables $z_1,z_2$ it remains as a
factors, hence it is sufficient to look at
\begin{equation}\label{interm-1}
\tilde K(z_{1},z_2;\bz_{\widehat
    {1,2}}) = \pi_1(u, h)\left(K(z_1;\bz_{\widehat 1})G^{(F_{h},h)}(\bz_{\widehat
  {2}};{\bf n}_{ h})\right)
\end{equation}
Since the function $\tilde K(z_1,z_2;\bz_{\widehat {1,2}})$ is
symmetric in the variables $z_{3},\dots, z_{h+1}$, when we proceed
with the action of the remaining divided difference operators
$\pi_h(u;1)\cdots \pi_2(u;h-1)$,  we are in the case $h-1$ and by
induction we get
\begin{equation}\label{interm-2}
\sum_{j=2}^{h+1}
\frac{\tilde K(z_1,z_j;\bz_{\widehat{1,j}})\prod_{i=1}^{ h-1}(z_j-\tau_i)}{\prod_{2
    \leq i\neq j\leq h+1}(z_j-z_i)}
\prod_{j=1}^{h-1}\frac{(z_j-\nu_u)}{(\tau_j-\nu_u)}
G^{(F_j,j)}(\bz_{\widehat{h-j+1}}).
\end{equation}
It remains to compute the sum in the previous equation.
For this we split $\tilde K(z_1,z_j;\bz_{\widehat {1,j}})$ in two parts
$$
\frac{(z_j-\tau_{h})(z_1-\nu_u) }{(\tau_{ h}-\nu_u)(z_1-z_j)} 
(K(z_1;\bz_{\widehat 1})
  G^{(F_h,h)}(\bz_{\widehat
    {j}})-K(z_j;\bz_{\widehat j})
  G^{(F_h,h)}(\bz_{\widehat
    {1}})).
$$
Once substituted into eq.(\ref{interm-2}) the leftmost term 
provides the terms for $2\leq j \leq h+1$ in the sum in eq.(\ref{eq-general-exch-prop}). 
For the first term we are led to consider the sum
$$
-\frac{(z_1-\nu_u)K(z_1;\bz_{\widehat 1})}{(\tau_{ h}-\nu_u)}\sum_{j=2}^{h+1}
\frac{ G^{(F_h,h)}(\bz_{\widehat
    {j}})\prod_{i=1}^{h}(z_j-\tau_i)}{\prod_{1
    \leq i\neq j\leq h+1}(z_j-z_i)},
$$
which can be easily evaluated using Proposition \ref{prop-sum-1},
giving the remaining term in the sum  in
eq.(\ref{eq-general-exch-prop}), namely
$$
\frac{(z_1-\nu_u)}{(\tau_{h}-\nu_u)}\frac{K(z_1;\bz_{\widehat
    1})\prod_{i=1}^h(z_1-\tau_i)}{\prod_{2\leq i \leq
    h+1}(z_1-z_i)} .
$$
\end{proof}

\bibliographystyle{amsplain}

\bibliography{bibliografia}

\providecommand{\bysame}{\leavevmode\hbox to3em{\hrulefill}\thinspace}
\providecommand{\MR}{\relax\ifhmode\unskip\space\fi MR }
\providecommand{\MRhref}[2]{%
  \href{http://www.ams.org/mathscinet-getitem?mr=#1}{#2}
}
\providecommand{\href}[2]{#2}
\begin{thebibliography}{10}

\bibitem{aas2015continuous}
E.~Aas and S.~Linusson, \emph{Continuous multi-line queues and the tasep},
  arXiv preprint arXiv:1501.04417 (2015).

\bibitem{aas2013product}
E.~Aas and J.~Sj{\"o}strand, \emph{A product formula for the tasep on a ring},
  arXiv preprint arXiv:1312.2493 (2013).

\bibitem{alcaraz1993reaction}
F.~C. Alcaraz and V.~Rittenberg, \emph{Reaction-diffusion processes as physical
  realizations of hecke algebras}, Physics Letters B \textbf{314} (1993),
  no.~3, 377--380.

\bibitem{arita2013matrix}
C.~Arita and K.~Mallick, \emph{Matrix product solution of an inhomogeneous
  multi-species tasep}, Journal of Physics A: Mathematical and Theoretical
  \textbf{46} (2013), no.~8, 085002.

\bibitem{ayyer2014correlations}
A.~Ayyer and S.~Linusson, \emph{Correlations in the multispecies tasep and a
  conjecture by lam}, arXiv preprint arXiv:1404.6679 (2014).

\bibitem{ayyer2014inhomogeneous}
\bysame, \emph{An inhomogeneous multispecies tasep on a ring}, Advances in
  Applied Mathematics \textbf{57} (2014), 21--43.

\bibitem{cantini-multispecies-2}
L.~Cantini, \emph{Inhomogenous multi-species tasep on a ring with spectral
  parameters 2: general contents}, To appear.

\bibitem{cantini2008algebraic}
\bysame, \emph{Algebraic bethe ansatz for the two species asep with different
  hopping rates}, Journal of Physics A: Mathematical and Theoretical
  \textbf{41} (2008), no.~9, 095001.

\bibitem{chou2011non}
T.~Chou, K.~Mallick, and R.~K.~P. Zia, \emph{Non-equilibrium statistical
  mechanics: from a paradigmatic model to biological transport}, Reports on
  progress in physics \textbf{74} (2011), no.~11, 116601.

\bibitem{corteel2007tableaux}
S.~Corteel and L.~K. Williams, \emph{Tableaux combinatorics for the asymmetric
  exclusion process}, Advances in applied mathematics \textbf{39} (2007),
  no.~3, 293--310.

\bibitem{corwin2012kardar}
I.~Corwin, \emph{The kardar--parisi--zhang equation and universality class},
  Random matrices: Theory and applications \textbf{1} (2012), no.~01, 1130001.

\bibitem{crampe2015new}
N.~Crampe, L.~Frappat, E.~Ragoucy, and M.~Vanicat, \emph{A new braid-like
  algebra for baxterisation}, arXiv preprint arXiv:1509.05516 (2015).

\bibitem{derrida1998exactly}
B.~Derrida, \emph{An exactly soluble non-equilibrium system: the asymmetric
  simple exclusion process}, Physics Reports \textbf{301} (1998), no.~1,
  65--83.

\bibitem{derrida1993exact}
B.~Derrida, M.~R. Evans, V.~Hakim, and V.~Pasquier, \emph{Exact solution of a
  1d asymmetric exclusion model using a matrix formulation}, Journal of Physics
  A: Mathematical and General \textbf{26} (1993), no.~7, 1493.

\bibitem{di2005around}
P.~Di~Francesco and P.~Zinn-Justin, \emph{Around the razumov-stroganov
  conjecture: proof of a multi-parameter sum rule}, Journal of Combinatorics
  \textbf{12} (2005), no.~1, R6.

\bibitem{essler1996representations}
F.~H.~L. Essler and V.~Rittenberg, \emph{Representations of the quadratic
  algebra and partially asymmetric diffusion with open boundaries}, Journal of
  Physics A: Mathematical and General \textbf{29} (1996), no.~13, 3375.

\bibitem{ferrari2005multiclass}
P.~A. Ferrari and J.~B. Martin, \emph{Multiclass processes, dual points and
  m/m/1 queues}, arXiv preprint math-ph/0509045 (2005).

\bibitem{ferrari2007stationary}
\bysame, \emph{Stationary distributions of multi-type totally asymmetric
  exclusion processes}, The Annals of Probability (2007), 807--832.

\bibitem{fomin1994grothendieck}
S.~Fomin and A.~N. Kirillov, \emph{Grothendieck polynomials and the yang-baxter
  equation}, Proc. Formal Power Series and Alg. Comb, 1994, pp.~183--190.

\bibitem{gwa1992bethe}
L.~H. Gwa and H.~Spohn, \emph{Bethe solution for the dynamical-scaling exponent
  of the noisy burgers equation}, Physical Review A \textbf{46} (1992), no.~2,
  844.

\bibitem{johansson2000shape}
K.~Johansson, \emph{Shape fluctuations and random matrices}, Communications in
  mathematical physics \textbf{209} (2000), no.~2, 437--476.

\bibitem{karimipour1999multispecies}
V.~Karimipour, \emph{Multispecies asymmetric simple exclusion process and its
  relation to traffic flow}, Physical Review E \textbf{59} (1999), no.~1, 205.

\bibitem{kuniba2015multispecies}
A.~Kuniba, S.~Maruyama, and M.~Okado, \emph{Multispecies tasep and
  combinatorial $r$}, Journal of Physics A: Mathematical and Theoretical
  \textbf{48} (2015), no.~34, 34FT02.

\bibitem{kuniba2016multispecies}
\bysame, \emph{Multispecies tasep and the tetrahedron equation}, Journal of
  Physics A: Mathematical and Theoretical \textbf{49} (2016), no.~11, 114001.

\bibitem{lam2012markov}
T.~Lam and L.~K. Williams, \emph{A markov chain on the symmetric group that is
  schubert positive?}, Experimental Mathematics \textbf{21} (2012), no.~2,
  189--192.

\bibitem{lascoux2003symmetric}
A.~Lascoux, \emph{Symmetric functions and combinatorial operators on
  polynomials}, vol.~99, American Mathematical Soc., 2003.

\bibitem{lascoux1983symmetry}
A.~Lascoux and M-P. Sch{\"u}tzenberger, \emph{Symmetry and flag manifolds},
  Invariant theory, Springer, 1983, pp.~118--144.

\bibitem{lascoux1985schubert}
\bysame, \emph{Schubert polynomials and the littlewood-richardson rule},
  Letters in mathematical physics \textbf{10} (1985), no.~2-3, 111--124.

\bibitem{macdonald1991notes}
I.~G. Macdonald, \emph{Notes on schubert polynomials}, vol.~6, Montr{\'e}al:
  D{\'e}p. de math{\'e}matique et d'informatique, Universit{\'e} du Qu{\'e}bec
  {\`a} Montr{\'e}al, 1991.

\bibitem{manivel2001symmetric}
L.~Manivel, \emph{Symmetric functions, schubert polynomials, and degeneracy
  loci}, no.~3, American Mathematical Soc., 2001.

\bibitem{rakos2005bethe}
A.~R{\'a}kos and G.~M. Sch{\"u}tz, \emph{Bethe ansatz and current distribution
  for the tasep with particle-dependent hopping rates}, arXiv preprint
  cond-mat/0506525 (2005).

\bibitem{rezakhanlou1991hydrodynamic}
F.~Rezakhanlou, \emph{Hydrodynamic limit for attractive particle systems on
  $\mathbb z^d$}, Communications in mathematical physics \textbf{140} (1991),
  no.~3, 417--448.

\bibitem{uchiyama2004asymmetric}
M.~Uchiyama, T.~Sasamoto, and M.~Wadati, \emph{Asymmetric simple exclusion
  process with open boundaries and askey--wilson polynomials}, Journal of
  Physics A: Mathematical and General \textbf{37} (2004), no.~18, 4985.

\end{thebibliography}

\end{document}